\documentclass[runningheads]{llncs}

\usepackage{amsmath}
\usepackage{amssymb}
\usepackage{latexsym}
\usepackage{cite}

\usepackage[pagewise]{lineno}

\usepackage{graphicx}

\usepackage[loose]{subfigure}

\usepackage[ruled,vlined,linesnumbered]{algorithm2e}

\newcommand{\CGP} {\ensuremath{\mathcal{A}}}
\newcommand{\NP} {NP}
\newcommand{\IC} {IC}
\newcommand{\NIC} {NIC}
\newcommand{\RAC} {RAC}
\newcommand{\Edge}[2]{\{#1,#2\}}

\begin{document}
\title{Recognizing \IC{}-Planar and  \NIC{}-Planar Graphs
\thanks{Supported by the Deutsche Forschungsgemeinschaft (DFG), grant
    Br835/18-1.}
}

\author{Franz J.\ Brandenburg}
\institute{%
University of Passau,
94030 Passau, Germany \\
\email{brandenb@fim.uni-passau.de}
}

\maketitle

\begin{abstract}
We prove that triangulated \IC{}-planar and \NIC{}-planar graphs can
be recognized in cubic time.
 A graph is 1-planar if it  can be drawn in the plane
 with at most one crossing per edge. A drawing is \IC{}-planar  if,
 in addition,
 each vertex is incident to at most one crossing edge  and \NIC{}-planar
  if two pairs of crossing edges share at
most one vertex. In a triangulated drawing  each face is a triangle.
In consequence, planar-maximal and maximal  \IC{}-planar and
\NIC{}-planar graphs can be recognized in $O(n^5)$ time and maximum
and optimal ones in $O(n^3)$ time. In contrast,
 recognizing 3-connected \IC{}-planar and \NIC{}-planar
graphs is \NP-complete, even if the graphs are given with a rotation
system which describes the cyclic ordering of the edges at each
vertex. Our results complement similar ones for 1-planar graphs.
\end{abstract}

\section{Introduction}
\label{sect:intro}

Graphs are commonly drawn in the plane  so that the vertices are
mapped to distinct points and the edges to Jordan curves connecting
the endpoints. A drawing is used to visualize structural
relationships that are modeled by vertices and edges and thereby
make them easier comprehensible to a human user. Specifications of
nice drawings of graphs and algorithms for their constructions are
the topic of Graph Drawing  \cite{dett-gdavg-99, kw-dg-01,
t-handbook-GD}.

There are several classes of graphs that are defined by specific
restrictions of edge crossings in graph drawings.  Edge crossings
are negatively correlated to nice, and therefore, they should be
avoided or controlled in some way.
 The  planar graphs are the
best known and most prominent example. Planarity excludes  crossings
 and is one of the most basic and influential concepts in Graph
Theory. Many properties of planar graphs have been explored,
including duality, minors, and drawings \cite{d-gt-00}, as well as
linear-time algorithms for the recognition and the construction of
straight-line grid drawings \cite{fpp-hdpgg-90, s-epgg-90}. However,
graphs from applications in engineering, social science, and life
science are generally not planar. This observation has motivated
approaches towards \emph{beyond-planar} graphs, which allow
crossings of edges with restrictions.
 A prominent example is 1-planar graphs, which were
introduced by Ringel \cite{ringel-65} in an approach to color a
planar graph and its dual simultaneously. A graph is \emph{1-planar}
if it can be drawn in the plane so that each edge is crossed at most
once. 1-planar graphs have found recent interest as emphasized by
Liotta's survey \cite{l-beyond-14}.
 A 1-planar graph
 of size $n$ has at most $4n-8$ edges \cite{ringel-65} and $K_6$ is the maximum complete
1-planar graph. 1-planar graphs do not admit straight-line drawings
\cite{t-rdg-88}, whereas 3-connected 1-planar graphs can be drawn
straight-line on a grid of quadratic size with the exception of a
single edge in the outer face \cite{abk-sld3c-13}.   Moreover,
1-planar graphs do not admit right angle crossing drawings
\cite{el-racg1p-13}, and conversely, there are right angle crossing
(\RAC{}) graphs  that are not 1-planar. In other words,
 the classes
of 1-planar  and \RAC{}   graphs  are incomparable. The recognition
problem of 1-planar graphs is \NP-complete \cite{GB-AGEFCE-07,
km-mo1ih-13}. It remains \NP-complete, even for graphs of bounded
bandwidth, pathwidth, or treewidth \cite{bce-pc1p-13}, if an edge is
added to a planar graph \cite{cm-aoepl-13}, and if the graphs are
3-connected and are given with a rotation system which describes the
cyclic ordering of the neighbors at each vertex \cite{abgr-1prs-15}.
On the other hand, 1-planar graphs can be recognized in cubic time
if they are triangulated \cite{cgp-rh4mg-06}  and even in linear
time if they are optimal and have $4n-8$ edges \cite{b-ro1plt-16}.

1-planar graphs can also be defined in terms of maps
\cite{cgp-pmg-98, cgp-mg-02, cgp-rh4mg-06, t-mgpt-98}. Maps
generalize the concept of planar duality. A \emph{map} $M$ is a
partition of the sphere into finitely many regions. Each region is
homeomorphic to a closed disk and the interior of two regions is
disjoint.   Some regions are labeled as \emph{countries} and the
remaining regions are lakes or \emph{holes}. In the plane, we use
the region of one country as outer face, which is unbounded and
encloses all other regions. An \emph{adjacency} is defined by a
touching of countries.  There is a \emph{strong} adjacency between
two countries if their boundaries intersect in a segment and a
\emph{weak} adjacency if the boundaries intersect only in a point.
There is a $k$-point if $k$ countries meet at a point. A map $M$
defines a graph $G$ so that the countries of $M$ are in one-to-one
correspondence with the vertices of $G$ and there is an edge $\{u,
v\}$ if and only if the countries of $u$ and $v$ are adjacent. Then
$G$ is called a \emph{map graph} and $M$ is the map of $G$. Note
that  holes are discarded for the definition of map graphs.
Obviously, a $k$-point induces $K_k$ as a subgraph. If no more than
$k$ countries meet at a point, then $M$ is a $k$-map and $G$ is a
$k$-\emph{map graph}.   If there are no holes then $M$ is
hole-free. A graph is  a \emph{hole-free 4-map graph} if it is the
map graph of a hole-free $4$-map \cite{cgp-pmg-98, cgp-mg-02,
cgp-rh4mg-06}.

Chen et al~\cite{cgp-mg-02, cgp-rh4mg-06} stated that the
triangulated 1-planar graphs are exactly the 3-connected hole-free
4-map graphs. Their fundamental result is a cubic-time recognition
algorithm for 3-connected hole-free 4-map graphs. They also observed
that the recognition problem of hole-free 4-map graphs can be
reduced in linear time to the special case of 3-connected graphs.
Hence, the recognition problem of 1-planar graphs is   solvable in
cubic time if the graphs are triangulated. We  extend the algorithm
to triangulated 1-planar graphs with (near) independent crossings.

A graph is \IC{}-planar (independent crossing planar)
\cite{a-cnircn-08,bdeklm-IC-16,ks-cpgIC-10,zl-spgic-13} if it has a
1-planar drawing   in which each vertex is incident to at most one
crossing edge  and is \NIC{}-planar (near independent crossing
planar) \cite{bbhnr-NIC-16,z-dcmgprc-14} if two pairs of crossing
edges share at most one vertex. If each pair of crossing edges is
augmented to the complete graph $K_4$, which is drawn as
 a kite as
in Fig.~\ref{fig:kite}, then a 1-planar drawing   is \IC-planar{} if
each vertex is
 part of at most one kite  and it is  \NIC{}-planar if each edge  is
 part of at most one kite.
It is known that \IC{}-planar graphs have at most $13/4 \, n -6$
edges \cite{ks-cpgIC-10} and are 5-colorable \cite{zl-spgic-13}. The
recognition problem is \NP-hard, even for 3-connected graphs with a
given rotation system \cite{bdeklm-IC-16}. \IC{}-planar graphs admit
straight-line drawings on a  grid of quadratic size and   right
angle crossing drawings, which, however,  may  need exponential area
\cite{bdeklm-IC-16}. Hence, every \IC-planar graph is a \RAC{}
graph. \NIC{}-planar graphs have at most $18/5 \, (n-2)$ edges
\cite{bbhnr-NIC-16, z-dcmgprc-14} and an \NP-complete recognition
problem. They admit straight-line drawings, but not necessarily with
right angle crossings. In fact, there are \NIC{}-planar graphs that
are not \RAC{} graphs, and vice-versa \cite{bbhnr-NIC-16}. Hence,
the classes of \NIC{}-planar graphs  and of \RAC{} graphs are
incomparable.
Outer 1-planar  graphs are another important subclass of 1-planar
 graphs
 that admit a 1-planar drawing with all vertices   in the outer
face \cite{e-rdg-86}. Outer 1-planar graphs are planar
\cite{abbghnr-o1p-15} and can be recognized  in linear time
\cite{abbghnr-o1p-15, heklss-ltao1p-15}.

A drawn graph defines an embedding which is an equivalence class of
drawings and consists of faces whose boundary consists of edges or
half-edges between a vertex and a crossing point of two edges. There
are several ways to augment 1-planar embeddings and graphs. Ringel
\cite{ringel-65} observed that each pair of crossing edges of a
1-planar embedding can be augmented to a complete graph $K_4$ that
is embedded as a kite. This fact has been rediscovered in many
works. A 1-planar embedding of  a graph $G$ is \emph{plane-maximal}
if no planar edge can be added to $G$ without violating 1-planarity
or introducing multiple edges. However, the introduction of multiple
edges may be useful at a separation pair. If there are two vertices
$s$ and $t$ so that  $G - \{s,t\} $ decomposes into components
$H_1,\ldots, H_r$, then the augmented components $H_i + \{s,t\}$ are
treated separately for a recognition \cite{cgp-rh4mg-06} or a
drawing \cite{b-vr1pg-14}. Altogether, there are $r-1$ copies of the
edge between the separation pair $\{s,t\}$.
 An embedding is \emph{triangulated} if each face is a
triangle and is bounded by three (half-)edges (between a vertex and
a crossing point). Clearly, a 1-planar embedding without separation
pairs is triangulated if and only if it is plane-maximal. If
multiple edges are added at a separation pair as described above,
then triangulated and plane-maximal coincide on 1-planar embeddings.

 A 1-planar graph $G$  is   triangulated (plane-maximal) if
it admits a  triangulated (plane-maximal) 1-planar embedding. It is
\emph{planar-maximal} if no  edge $e$ can be added to $G$ so that
$G+e$ admits a 1-planar embedding in which $e$ is planar. Finally,
$G$
 is \emph{maximal}  if $G+e$ is not  1-planar,   \emph{maximum} or \emph{densest} if
$G+e$ violates the upper bound of the number of edges of 1-planar
graphs and \emph{optimal} if the number of edges exactly meets the
upper bound of $4n-8$. Hence, a graph in a graph class $\mathcal{G}$
is maximal if there is no supergraph in $\mathcal{G}$ with the same
set of vertices and a proper superset of edges, and maximum if there
is no graph in $\mathcal{G}$ of the same size and with more edges.
Similar notions apply to planar, \IC{}-planar, and \NIC{}-planar
graphs. Clearly, these concepts coincide for planar graphs whereas
they differ for \IC{}-planar, \NIC{}-planar and 1-planar graphs.
First, note the difference between plane-maximal and planar-maximal
graphs. As an example, remove an edge from the complete graph on
five vertices and consider $K_5-e$ which is a maximal planar graph.
Every planar embedding of $K_5-e$ is plane-maximal 1-planar.
However, the removed edge $e$ can be added and drawn planar if a
$K_4$ subgraph of $K_5-e$ is drawn with a pair of crossing edges.
Hence, $K_5-e$ is plane-maximal 1-planar and not planar-maximal
1-planar. Similarly, every triangulated planar graph of size at
least five is plane-maximal and not planar-maximal (or maximal)
1-planar.
 Bodendiek et al. \cite{bsw-bs-83} showed
that densest 1-planar graphs have $4n-8$ edges and that such graphs,
called optimal, exist for $n=8$ and all $n \geq 10$
\cite{bsw-1og-84}. The upper bound was rediscovered in many works.
Bodendiek et al. also observed that there are maximal 1-planar
graphs that are not optimal. The gap in the number of edges of
maximal 1-planar is quite large, as shown by Brandenburg et al.
\cite{begghr-odm1p-13}, who found sparse maximal 1-planar graphs
with  $45/17 \, n - 84/17$ edges. Similarly, there are  sparse
maximal \IC{}-planar graphs with $3n-5$ edges and sparse maximal
\NIC{}-planar graphs with $16/5 \, (n-2)$ edges, and both bounds are
tight \cite{bbhnr-NIC-16}. There are optimal \IC{}-planar graphs
only for  $n=4k$  and optimal \NIC{}-planar graphs only for $n=5t+2$
and such graphs exist for all $k \geq 2$ \cite{zl-spgic-13} and all
$t \geq 2$ \cite{bbhnr-NIC-16}. Maximum \IC{}-planar graphs with
$\lfloor 13/4 \,n-6 \rfloor$ edges exist for all $n \geq 5$ and
there are maximum \NIC{}-planar graphs with $\lfloor 18/5(n-2)
\rfloor$ edges for $n=5t+i$ and $i=2,3$ \cite{bbhnr-NIC-16}. Hence,
the sequence  of restrictions from triangulated to optimal is proper
for 1-planar, \IC{}-planar, and \NIC{}-planar graphs.

Finally, note that triangulated  \IC-planar (\NIC{}-planar)
embeddings do not admit separation pairs so that the embeddings are
in normal form with a kite at each pair  of crossing edges
\cite{abk-sld3c-13} and planar triangles for the other faces. There
is a planar generalized dual graph if each kite (and also each
planar tetrahedron) is represented by a special node
\cite{bbhnr-NIC-16}.

In this work we extend the cubic-time algorithm of Chen et
al.~\cite{cgp-rh4mg-06} for the recognition of   triangulated
1-planar graphs to triangulated \IC{}-planar and \NIC{}-planar
graphs. We call the  algorithms \CGP{}, $\mathcal{B}_{IC}$ and
$\mathcal{B}_{NIC}$, respectively. Our algorithms are presented as a
program and consist of three parts. They compute an edge coloring
and a boolean formula which is used to test \IC{}- and
\NIC{}-planarity.

The paper is organized as follows. Section \ref{sect:basics}
describes basic definitions. In Section \ref{sect:recog} we present
our algorithm and we show how to solve \IC{}- and \NIC{}-planarity
in Section \ref{sect:IC+NIC}.
 We conclude in Section \ref{sect:conclusion} with some
open problems.

\section{Preliminaries}
\label{sect:basics}
We consider undirected graphs $G = (V, E)$  and   assume that the
graphs are simple and 2-connected, unless otherwise stated. The
subgraph induced by a subset $U$ of vertices is denoted by $G[U]$.
For convenience, we omit braces and write $G[u_1, \ldots, u_r]$ if
$U = \{u_1, \ldots, u_r\}$. The subgraph of $G$ induced by the
vertices of subgraphs $H$ and $K$ is denoted by $H+K$, and similarly
for $G-H$, except if $H$ is an edge, which is removed from $G-H$
whereas the endvertices remain.

A  \emph{drawing} of $G$ maps  the vertices  to distinct points in
the plane and each edge $\{u,v\}$   to a Jordan arc connecting the
points of $u$ and $v$.  Two edges \emph{cross} if their Jordan arcs
intersect. A planar drawing excludes edge crossings and a 1-planar
drawing admits at most one crossing per edge. A crossings subdivides
an edge into two \emph{half-edges}. An  \emph{embedding}
$\mathcal{E}(G)$ is an equivalence class of drawings and specifies
edge crossings and faces.
 The planarization of an embedding $\mathcal{E}(G)$ is an
embedded planar graph which is obtained by taking each crossing
point  as a new vertex and half-edges as new edges.

A  planar embedding   partitions the plane (or the sphere) into
\emph{faces} or \emph{regions},  called a map  \cite{cgp-mg-02,
cgp-rh4mg-06}.  Two faces are \emph{adjacent} if their boundaries
intersect. The intersection is a segment or just a common point. A
\emph{hole-free map graph} is defined by a one-to-one correspondence
between faces and vertices and between adjacencies and edges. There
is a $k$-map graph if at most $k$ regions meet at a point of a map.

\begin{figure}
  \centering
  \subfigure[]{
    \includegraphics[scale=0.30]{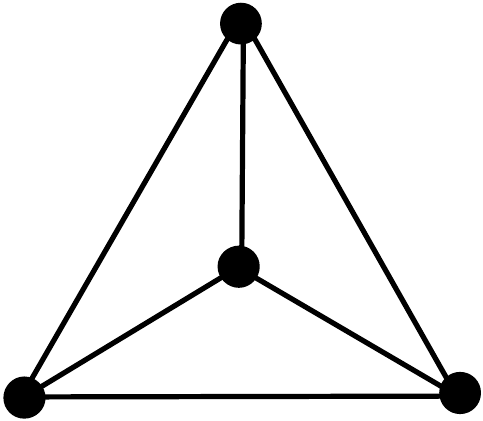}
    \label{fig:tetrahedron}
  }
  \hfil
  \subfigure[]{
      \includegraphics[scale=0.27]{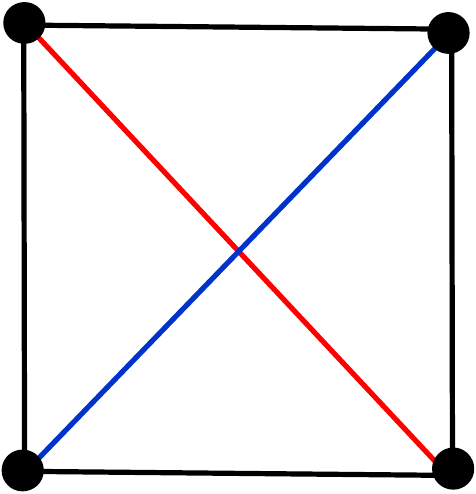}
      \label{fig:kite}
  }
  \caption{Drawings of $K_4$  \subref{fig:tetrahedron} planar as a
    tetrahedron and \subref{fig:kite} with a crossing as a kite.}
  \label{fig:k4-embeddings}
\end{figure}

The complete graph on four vertices $K_4$ plays a crucial role in
1-planar, \IC{}-planar and \NIC{}-planar  graphs. It admits two
embeddings \cite{Kyncl-09}, as a \emph{tetrahedron} or as a
\emph{kite} with a pair of crossing edges, see
Fig.~\ref{fig:k4-embeddings}. The embedding as a tetrahedron is not
necessarily planar. A planar edge can be \emph{covered} by a kite so
that an edge   is a crossing edge of a kite, see
Figs.~\ref{fig:separatingtriple} and \ref{fig:separatingtriangle}.
The cubic-time recognition algorithm for hole-free 4-map graphs of
Chen et al. \cite{cgp-rh4mg-06} searches all $K_4$ subgraphs
$\kappa$ of the given graph and checks whether $\kappa$ must be
embedded as a kite or as a tetrahedron. This can be determined to a
large extend, but it is not unique, as a $K_5$ illustrates. The
complete graph $K_5$ has five embeddings (up to graph automorphism)
\cite{Kyncl-09} as displayed in Fig.~\ref{fig:allK5}, but only one
of them is 1-planar. If the outer face is fixed, then there are
three 1-planar embeddings with one of the outer edges in a kite, see
Fig.~\ref{fig:K5fixedface}.


\begin{figure}
   \centering
   \subfigure[ ]{
     \includegraphics[scale=0.4]{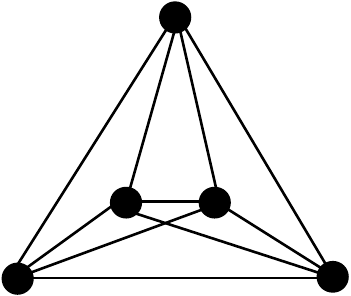}
   }
   \quad
      \subfigure[ ]{
     \includegraphics[scale=0.4]{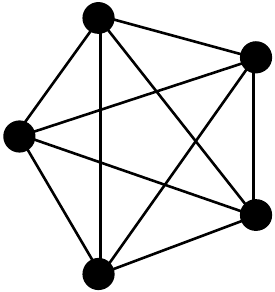}
   }
   \quad
      \subfigure[ ]{
     \includegraphics[scale=0.4]{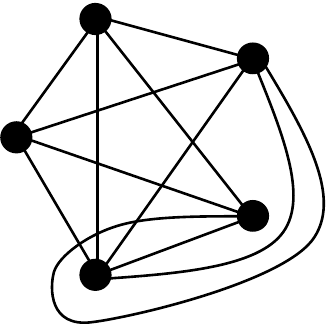}
   }
   \quad
      \subfigure[ ]{
     \includegraphics[scale=0.4]{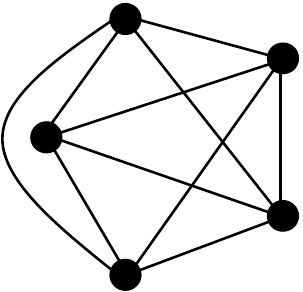}
   }
\quad
   \subfigure[ ]{
     \includegraphics[scale=0.4]{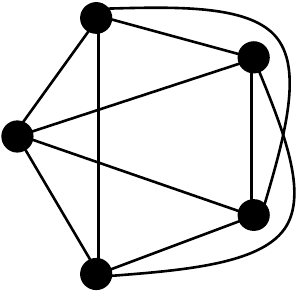}
   }
   \caption{All non-isomorphic embeddings of $K_5$}
   \label{fig:allK5}
\end{figure}

\begin{figure}
   \centering
   \subfigure[ ]{
     \includegraphics[scale=0.27]{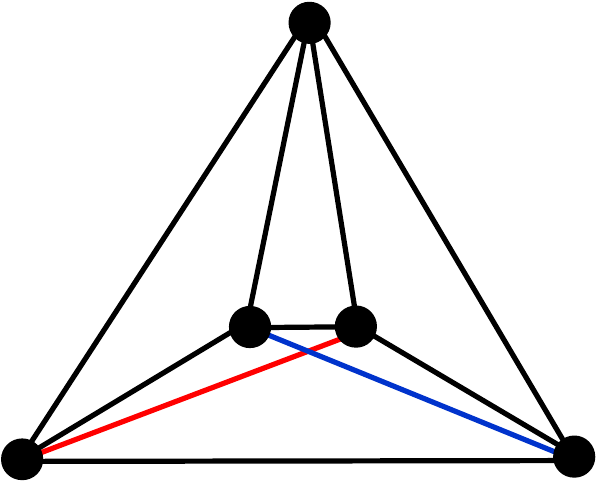}
   }
   \quad
      \subfigure[ ]{
     \includegraphics[scale=0.27]{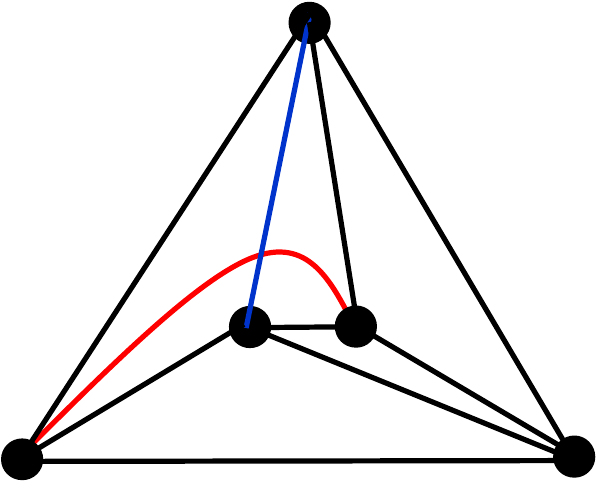}
   }
\quad
   \subfigure[ ]{
     \includegraphics[scale=0.27]{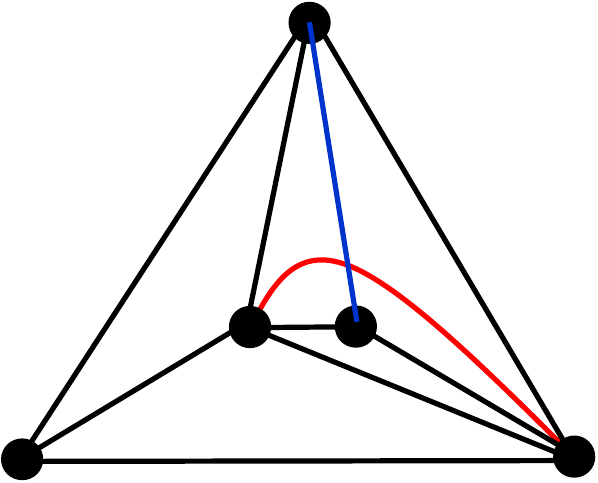}
   }
   \caption{Three embeddings of $K_5$ with a fixed outer face. Each kite includes the edge
   between the inner vertices and one of the outer edges.}
   \label{fig:K5fixedface}
\end{figure}


\section{Recognition} \label{sect:recog}

For the recognition of triangulated \IC{}-planar and \NIC{}-planar
graphs we extend algorithm \CGP{} of Chen et al.
\cite{cgp-rh4mg-06}. Recall that the 3-connected hole-free 4-map
graphs are exactly the triangulated 1-planar graphs. Our algorithm
$\mathcal{B}$ extends
 \CGP{} by an edge coloring and a boolean
formula. Algorithm \CGP{} marks an
edge if it is planar at the actual stage of the algorithm. A marked
edge could have been crossed at an earlier stage, in which case it
is crossed in the computed 1-planar embedding. This divergence is
due to the fact that  algorithm \CGP{} removes one crossing edge if it detects a pair of
crossing edges. The remaining
crossed edge is marked and is treated as planar. Our edge coloring records each
decision and  tells whether an edge is planar or crossed in
every triangulated 1-planar embedding,
  or whether this is uncertain and depends on a
particular embedding. The uncertainty is expressed by a boolean
formula such that there is a one-to-one correspondence between
feasible embeddings and truth assignments. An embedding is feasible
if it is triangulated and \IC{}- and \NIC{}-planar, respectively. We
prove the following result:

\begin{theorem}
There is a cubic-time algorithm   that checks whether a graph $G$ is
a triangulated 1-planar graph. It returns an edge coloring of $G$,
from which one obtains a partial embedding of a spanning subgraph of
$G$,
 and a boolean formula $\eta$ such that the
\IC{}-extension (\NIC{}-extension) $\eta^+$ of $\eta$ is satisfiable
if and only if the embedding  of $G$ is \IC{}-planar
(\NIC{}-planar). Otherwise, the algorithm returns \textsf{false}
and stops with a failure.
\end{theorem}

Algorithm $\mathcal{B}$ is the program of algorithm \CGP{} of
\cite{cgp-rh4mg-06} with a minor simplification. Algorithms
$\mathcal{B}_{IC}$ and  $\mathcal{B}_{NIC}$ specialize $\mathcal{B}$
to  \IC{}-planar and \NIC{}-planar graphs, respectively. They stop
immediately if there is a violation of \IC{}- or \NIC{}-planarity.
 In each step the algorithms  add a clause to a
CNF formula $\eta$. The boolean formulas for \IC{}-planar and
\NIC{}-planar graphs have the same structure, however, the boolean
variables and the evaluation are different. A boolean variable  is
associated with a vertex for \IC{}-planar graphs and with an edge
for \NIC{}-planar graphs. For every $K_4$ subgraph $\kappa$ of the
input graph $G$, the clause $\alpha(\kappa) = (a_{\kappa} \wedge
b_{\kappa} \wedge c_{\kappa} \wedge d_{\kappa})$ expresses that
 $\kappa$ is embedded as a kite with a pair of crossing edges and the boolean variable $x_{\kappa}$
is assigned the value \textsf{true}. Here, $x \in \{a,b,c,d\}$ is
a vertex of $\kappa$ in the \IC{}-planar case, and $x$ is a
planar edge of $\kappa$ in the \NIC{}-planar case.
Feasibility is granted by \IC{}- and \NIC{}-extensions of the form
$(\neg x_{\kappa} \vee \neg x_{\kappa'})$  for every vertex (edge)  $x$ and
kites $\kappa$ and $\kappa'$ that may  include $x$.

Algorithm \CGP{}  systematically checks all $K_4$ subgraphs $\kappa$
of the given input graph $G$. It checks  gadgets in a fixed order
and tries to determine whether $\kappa$ must be embedded as a kite
or as a tetrahedron, and so do algorithms $\mathcal{B}$,
$\mathcal{B}_{IC}$  and $\mathcal{B}_{NIC}$. In most cases there is
an unambiguous decision. However, a separating edge only tells that
there is a kite, but it does not fix its position. Another ambiguity
comes from small graphs which  result from a partition by a
separating 3-cycle or 4-cycle. Input graphs of size at most eight
are checked by inspection. For example, the complete graph  $K_5$
with a fixed outer face has three embeddings with kites $\kappa_1$,
$\kappa_2$, and $\kappa_3$,
  as illustrated in Fig.~\ref{fig:K5fixedface}, and the
clause $ \alpha(\kappa_1)  \vee \alpha(\kappa_2) \vee
\alpha(\kappa_3)$ expresses the three options.  For an efficient
evaluation the clause is simplified.

The specialization  concerns the steps of \CGP{} with a separating
triangle, $MC_5$ and $MC_4$. These gadgets are described below.  A
vertex is incident to two kites at a separating triangle, which
violates \IC{}-planarity. The $MC_5$ step searches $K_5$ and is
vacuous, even for 1-planar graphs, as proved in Lemma
\ref{lem:no-MC5}, and most subcases of the $MC_4$ step violate
\IC{}- and \NIC-planarity. Then algorithms $\mathcal{B}_{IC}$  and
$\mathcal{B}_{NIC}$ return \textsf{false} and   stop. In addition,
the set of all \IC{}-planar (\NIC{}-planar) embeddings of graphs of
size at most eight is computed and expressed by a simplified boolean
formula with variables for the vertices or edges of the outer face.

Next, we explain the edge coloring  and the boolean formula.

\begin{definition}
  An edge of a 1-planar graph $G$ is colored \emph{black} if it is
  a planar edge in every triangulated 1-planar embedding $\mathcal{E}({G})$. Two edges
  $e$ and $f$ are colored \emph{red} and \emph{blue}, respectively, if $e$
  and $f$ cross in every triangulated 1-planar embedding. In edge is colored
  \emph{orange} if it is crossed and the candidates for a crossing are
  colored \emph{cyan}. Finally, \emph{grey} edges are unclear.

  For an uncolored edge $\{a,b\}$  let $\mathcal{C}[a, b]$ be the set of
  uncolored edges $\{x,y\}$ so that the induced subgraph $G[a, b, x, y]$
  is a $K_4$. Such edges are called \emph{crossable edges} in
  \emph{\cite{cgp-rh4mg-06}}.
\end{definition}

The black, red, blue, and orange edges are decided, whereas a cyan
and a grey edge may be planar in one 1-planar embedding and crossed
in another. Grey edges appear only in small subgraphs, such as
$K_5$ in Fig.~\ref{fig:K5fixedface}. A partial coloring $\chi$ on a
subset of edges of $G$ is extended step by step such that some
uncolored edges are colored. Here, blue and cyan overrule black or
grey and blue and cyan edges keep their color. There is an error,
otherwise, e.g., if a black edge shall be colored blue or red.

\begin{definition}
  A \emph{coloring} $\gamma$ of a set of edges $F \subseteq E$ extends a
  partial edge coloring $\chi$ of $G$ if colored edges keep their color
  and an uncolored edge  $e \in F$ takes the color of $\gamma$. If $\gamma$
  and $\chi$ disagree,  then
  $\gamma(e)$ = ``black'' or ``grey'' and  $\chi(e)$ = ``blue'' or
  ``cyan''. Otherwise, there is a conflict between $\gamma$
  and $\chi$, which is reported as a failure.
\end{definition}

The boolean formula $\eta$ is build up step by step as a conjunction
of clauses during a run of our algorithms. For every $K_4$ subgraph
$\kappa = G[ a,b,c,d]$, which may be embedded as a kite, let
$\alpha(\kappa) = (x^1_{\kappa}  \wedge x^2_{\kappa} \wedge
x^3_{\kappa}  \wedge x^4_{\kappa})$. In case of \IC{}-planarity,
$x^1, \ldots, x^4$ are the vertices of $\kappa$, and they are the
planar edges of $\kappa$ in case of \NIC{}-planarity. Each vertex
(edge) may be part of at most one kite such that each variable
$x_{\kappa}$ is \textsf{true} for at most one kite $\kappa$. This is
expressed by the \IC{}- and \NIC{}-extension, respectively.\\

Let's recall algorithm  \CGP{} on triangulated 1-planar graphs which
are 3-connected hole-free map graphs. For the formal properties of
each step, its computation, and the correctness proof we refer to
\cite{cgp-rh4mg-06}. Algorithm \CGP{} ``makes progress'' (i) by a
separation and (ii) by a crossing removal.

There are separating $3$-cycles and separating $4$-cycles in $G$.
Each such cycle $C$ partitions $G - C$ into an inner and an outer
component $G_{in}$ and $G_{out}$. At the time of the separation, the
edges of $C$ are planar and the subgraphs $G_{in}+C$ and $G_{out}+C$
can be treated separately \cite{cgp-rh4mg-06}. If $C$ is a 4-cycle,
then a chord $f$ must be added to the subgraphs for a triangulation,
and $f$ must be chosen properly, such that it is new for the
remaining subgraph. The chord is removed if the subgraphs are merged
later on. Some edges of $C$ were kite-covered before and were
crossed by other edges that were removed in an earlier step of the
algorithm. Then some edges of $C$ are colored blue or cyan, whereas
 \CGP{}   treats  them as black edges.

Algorithm \CGP{} recursively searches for gadgets, namely,
separating 3-cycles, separating edges, separating 4-cycles,
separating triples, separating quadruples, separating triangles,
$MC_5$  and $MC_4$, in this order.  The gadgets are described below.
Hence, if \CGP{} considers a separating edge, then there are no
separating 3-cycles and the graph under consideration is
4-connected. There is neither of the other gadgets if   $MC_4$   is
applied. The search for the gadgets in the given order simplifies
the case analysis and implies that decisions hold for all
triangulated 1-planar embeddings.

In each case, algorithm \CGP{} finds edges that are crossed in a
1-planar embedding or finds  edges that can be treated as planar at this stage.
One edge from a pair of crossing edges is removed to make progress
towards planarity. If \CGP{} does not fail, then it terminates at a
triangulated planar graph or at a small graph of size at most eight.

The decisions of \CGP{} do not uniquely determine a 1-planar
embedding. For example, if there is a separating edge $e$,   then
$e$ has the choice among several crossable edges. The case resembles
a graph decomposition at a separation pair. Also,  $K_5$ has three
embeddings if there is a planar outer 3-cycle. A final ambiguity
comes from each   pair  of crossing edges where \CGP{} removes one
of them to make progress. The choice has an effect on the ongoing
computation process. One may aim at using separating 3- and 4-cycles
in the next step by removing the edges of $\mathcal{C}[a, b]$
whereas the removal of  $\{a,b\}$ aims at 4-connected planar graphs.
The choice does not affect the ``yes'' or ``no'' decision on a
triangulated 1-planar graph.

\IC{}-  and  \NIC{}-planarity need  more information on all 1-planar
embeddings of the given graph which is provided by an edge coloring
and a boolean formula.

\begin{definition}
Let $G$ be a 4-connected graph.
  \begin{enumerate}
  \item A \emph{separating edge}   is an uncolored
      edge $\{a,b\}$ such that $G-\{a,b\}- \mathcal{C}[a,b]$ is
    disconnected, see Fig. \ref{fig:separatingedge1}.
  \item A \emph{separating 4-cycle} $C = (a,b,c,d)$ is a 4-cycle such that
    $G-C$ is disconnected.
  \item A \emph{separating triple} is a 3-cycle $C =(a,b,c)$ such that $G-
    C- \mathcal{C}[a,b]$ is disconnected, see Fig.
    \ref{fig:separatingtriple}.
  \item A \emph{separating quadruple} $C = (a,b,c,d)$ is a 4-cycle such that
    $G-C-\mathcal{C}[a,b]$ is disconnected.
  \item A \emph{separating triangle} is a 3-cycle $C=(a,b,c)$ such that $G-
    C- \mathcal{C}[a,b]-\mathcal{C}[b,c]$ is disconnected, see
    Fig. \ref{fig:separatingtriangle}.
  \end{enumerate}
\end{definition}

\begin{figure}
   \begin{center}
     \includegraphics[scale=0.4]{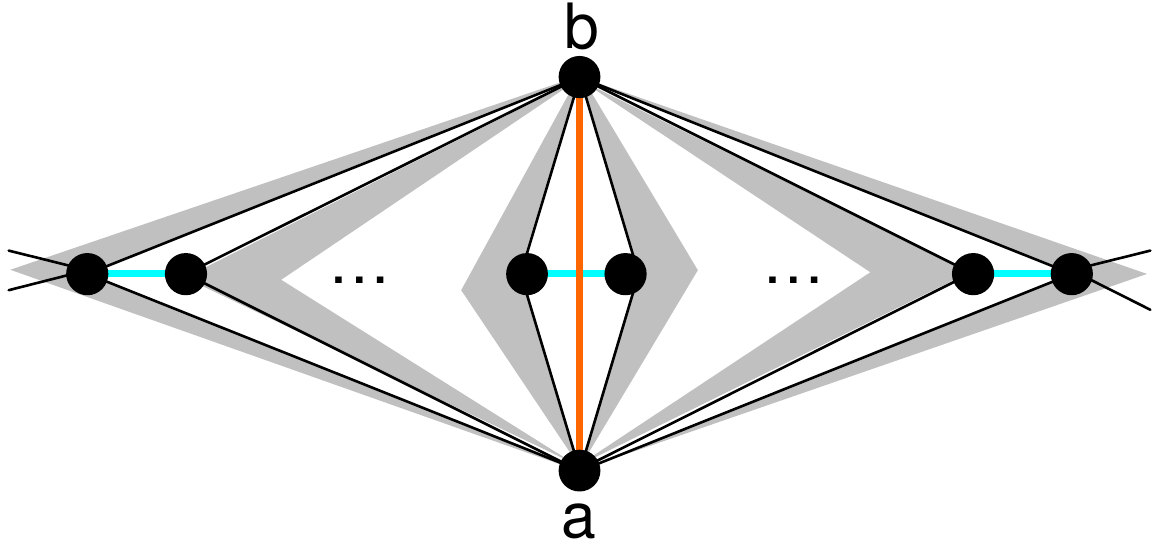}
     \caption{A separating edge $\{a,b\}$; the shaded areas represent
     4-connected
     subgraphs. The orange edge $\{a,b\}$ must cross one of the cyan
     ones.
     \label{fig:separatingedge1}}
   \end{center}
\end{figure}

\begin{figure}
  \centering
  \subfigure[] {
     \includegraphics[scale=0.40]{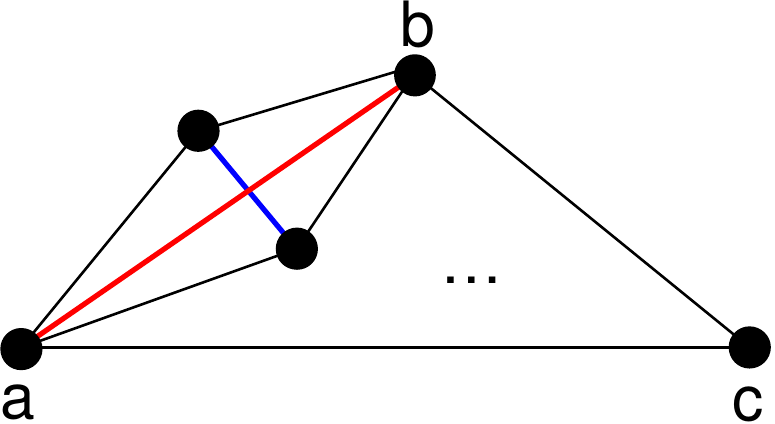}
    \label{fig:separatingtriple}
  }
  \hfil
  \subfigure[] {
      \includegraphics[scale=0.4]{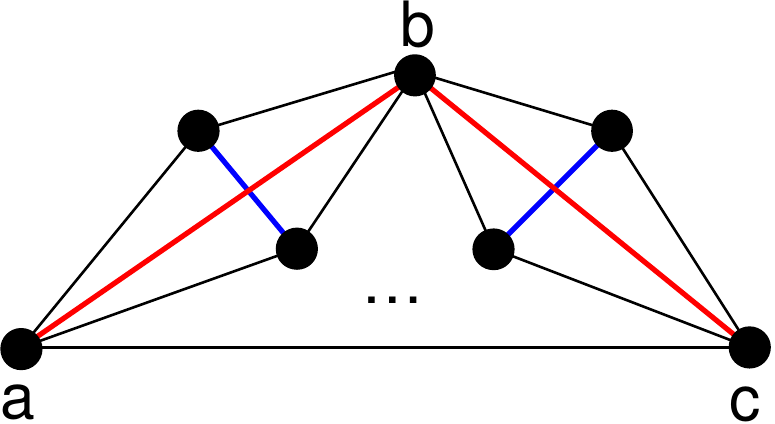}
      \label{fig:separatingtriangle}
  }
  \caption{A separating triple with a kite-covered edge $\{a,b\}$ and a separating
    triangle with two kite-covered edges. The dots represent a subgraph. }
  \label{fig:separating}
\end{figure}

We use the following properties of algorithm \CGP{}.

\begin{lemma} \label{propertiesA}
  Let $G$ be a triangulated 1-planar graph with $|G| > 8$.
 \begin{enumerate}
 \item The edges of separating 3-cycles and 4-cycles are planar at
 the time of their detection.
 \item If $\{a,b\}$ is a  separating edge, then $\{a,b\}$ is crossed
in every  1-planar embedding and the edges $\{a,x\}$ and $\{b,x\}$
are planar for each vertex $x$ with a crossable edge $\{x,y\} \in
\mathcal{C}[a,b]$.
 \item Edge $\{a, b\}$ of a  separating triple
  (quadruple) is crossed in every
   1-planar embedding, and similarly  edges $\{a, b\}$ and
   $\{b, c\}$ of a separating triangle.
 \item If $\{u,v\}$ is a   crossable edge of a separating triple,
 separating quadruple, and separating triangle with edge $\{a,
 b\}$, respectively, then $\{u,v\}$ is crossed in every 1-planar
 embedding whereas the edges $\{a,u\}, \{a,v\}, \{b,u\},
 \{b,v\}$ are planar in every 1-planar embedding.
 \end{enumerate}
\end{lemma}

\begin{proof}
Algorithm \CGP{}  partitions $G$ into an inner and an outer
component at a separating 3- or 4-cycle $C$ and marks the edges of
$C$ (Lemma 3.5 of \cite{cgp-rh4mg-06}). A marked edge is treated as
planar and is not considered for separating edges, triples,
quadruples, or triangles. Accordingly, if \CGP{} encounters a
separating edge $\{a,b\}$ and there is a crossable edge $\{x,y\}$
such that $x$ and $y$ belong to different connected components, then
the subgraph $G[a,b,x,y]$ can be embedded as a kite so that edge
$\{a,b\}$ crosses $\{x,y\}$ (Lemma 7.2 of \cite{cgp-rh4mg-06}). Edge
$\{a,b\}$  is crossed by one of the crossable edges of
$\mathcal{C}[a,b]$. The make progress step (remove a correct pizza
in Definition 5.1 of \cite{cgp-rh4mg-06}) removes $\{a,b\}$ or
$\{x,y\}$ and marks the edges $\{a,x\}$ and $\{b,x\}$ for all
vertices $x$ with a crossable edge $\{x,y\} \in \mathcal{C}[a,b]$.
Similarly, algorithm \CGP{}   proceeds for separating triples,
separating quadruples, and separating triangles.
%
\end{proof}

The search for maximal complete subgraphs of size five and  $MC_4$
 complete   algorithm \CGP{}.  However,  as stated before, $K_5$
subgraphs have a unique 1-planar embedding which is detected at an
earlier stage.

\begin{lemma} \label{lem:no-MC5}
 The $MC_5$ step of algorithm \CGP{} is  vacuous if $G$ is a
 triangulated 1-planar graph.
 \end{lemma}

\begin{proof}
There are five embeddings of $K_5$  \cite{Kyncl-09}  and only the
one in Fig.~\ref{fig:allK5}(a) is 1-planar. The embedding consist of
a kite and a top vertex $t$ (called crust in \cite{cgp-rh4mg-06}).
The edges incident with $t$ can be planar or are kite-covered,
whereas the outer edges of the kite are planar. Hence, there  is   a
separating 3-cycle,
 triple, or triangle and \CGP{}  takes these gadgets
with higher priority than $MC_5$.
%
\end{proof}

In consequence, we must consider embeddings of $K_5$ only as part of
a small subgraph $H$ of size at most eight that  is obtained
by a partition of a separating 3- or 4-cycle. The outer edges of $H$
are treated as planar although they may be colored black, blue or
cyan. For example, suppose there is a separating triangle $C =
\{a,b,c\}$ and $\{a,b\}$ is crossed by $\{x,y\}$. If edge $\{x,y\}$
is removed, then   $\{a,b\}$ is colored
blue and the edges $\{a,x\}, \{a,y\}, \{b,x\}, \{b,y\}$ are colored
black and are planar. If $G_{in}$ is a small subgraph obtained from
$G-C$ and $x$ is in $G_{in}$, then it has a  planar outer triangle
with vertices $a,b$ and $x$.

Algorithm \CGP{} finally applies   $MC_4$  and checks whether the
detected $K_4$ must be embedded as a tetrahedron or as a kite.
However, at this stage of the algorithm, the embedding as a
tetrahedron implies that all edges are kite-covered. A planar
tetrahedron is detected in the first step since there is a
separating 3-cycle and there is a   separating triangle or a
separating triple, otherwise.  Hence, only three cases remain, as
described in Section 9.1 of \cite{cgp-rh4mg-06}.

\begin{lemma} \label{lem:MC4}
If   $MC_4$  applies to algorithm  \CGP{}
then the $K_4$ subgraph is
\begin{enumerate}
  \item  a completely kite-covered tetrahedron, see
  Fig.~\ref{fig:riceball}
  \item  an $SC$-graph, see Fig.~\ref{fig:SCgraph}, or
  \item a kite, see Fig.~\ref{fig:kite}, \\
\noindent and they are checked in this order.
\end{enumerate}
\end{lemma}

\begin{figure}
  \centering
  \subfigure[] {
      \includegraphics[scale=0.35]{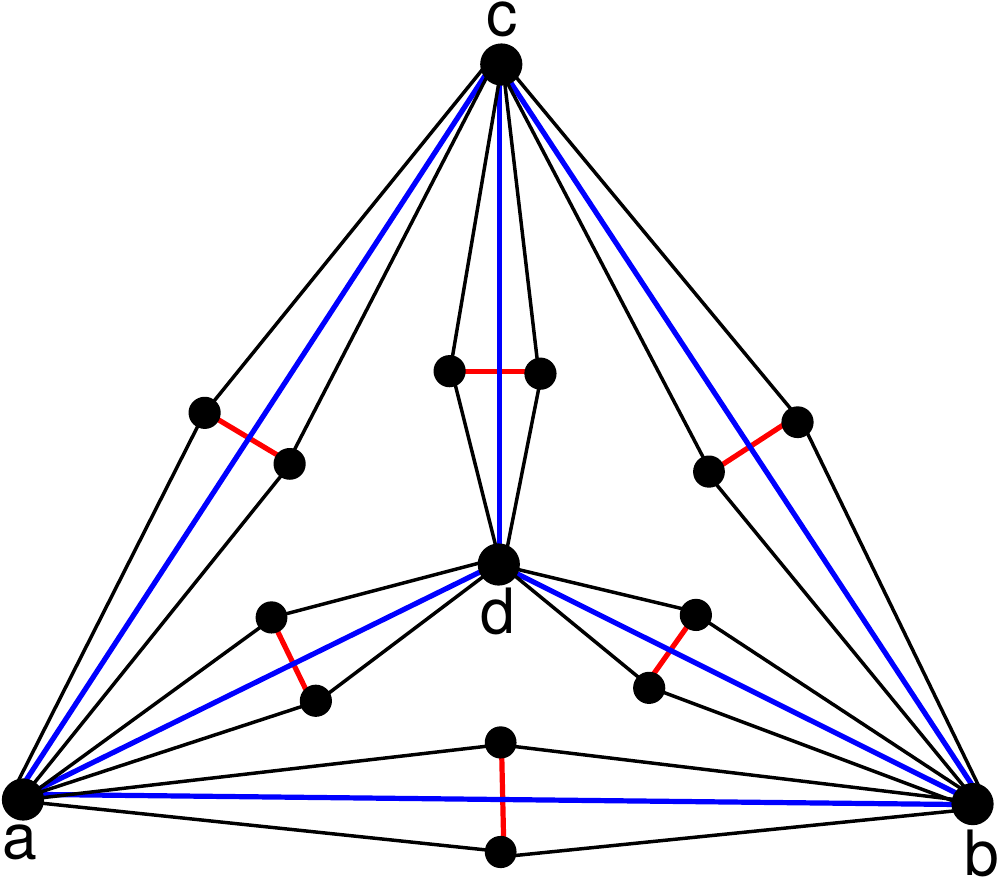}
    \label{fig:riceball}
  }
  \hfil
  \subfigure[] {
       \includegraphics[scale=0.35]{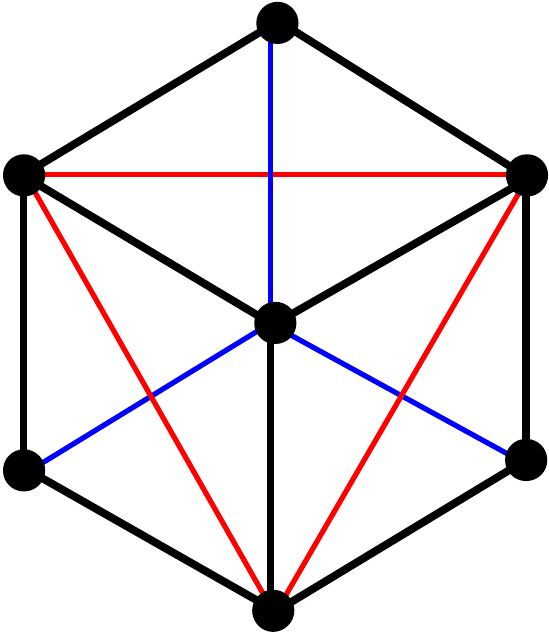}
      \label{fig:SCgraph}
  }
  \caption{A (a) completely kite-covered tetrahedron and (b) an SC-graph }
  \label{fig:MC4rule}
\end{figure}

A crossing edge  of each of the six kites it removed to make progress if
 a completely  kite-covered tetrahedron is detected. The $SC$-graph
plays a central in the graph reduction system of Schumacher
\cite{s-s1pg-86} and Brandenburg \cite{b-ro1plt-16}. Here a crossing
edge is removed from each of the three kites. Plain kites that are
surrounded by planar subgraphs are common in \IC{}-planar and in
\NIC{}-planar graphs, as the analysis of maximal \NIC{}-planar
graphs shows \cite{bbhnr-NIC-16}.

\begin{algorithm}
  \KwIn{A 3-connected graph $G$ with a partial edge coloring  and a boolean
        formula $\eta$. Initially, all edges are uncolored and
        $\eta = \textsf{true}$.}
  \KwOut{A planar embedding of an induced subgraph of $G$, an edge coloring of $G$,  and (an extension of) $\eta$.}
  \BlankLine
  \While{there is a $K_4$ subgraph and $|G| \geq 9$}{
    \If{there is a separating 3-cycle $C = (a, b, c)$ with \linebreak
          $G - C = \{G_{in}, G_{out}\}$}{
      extend the coloring by black  edges for $\Edge{a}{b}, \Edge{b}{c}, \Edge{c}{a}$\;
      \Return \textbf{merge}($\mathcal{B}(G_{in} + C, \eta_{in}), \,
                     \mathcal{B}(G_{out} + C, \eta_{out}))$\;
    }
    \ElseIf{there is a separating edge $\Edge{a}{b}$}{
      extend the coloring by an orange edge $\Edge{a}{b}$\;
      \ForEach{edge $\Edge{x}{y} \in \mathcal{C}[a, b]$}{
        extend the coloring by black edges $\Edge{a}{x}, \Edge{x}{b},\Edge{b}{y}, \Edge{a}{y}$  and a cyan edge
        $\Edge{x}{y}$\;
        add $\sigma(a, b, \mathcal{C}[a, b])$ to $\eta$\;
      }
      \Return $\mathcal{B}(G - \Edge{a}{b}, \eta)$\;
    }
    \ElseIf{there is a separating 4-cycle $C = (a, b, c, d)$ and \linebreak
          $G - C = \{G_{in}, G_{out}\}$}{
      extend the coloring by black  edges   $\Edge{a}{b},\Edge{b}{c},\Edge{c}{d},\Edge{d}{a}$\;
      \leIf{$\Edge{a}{c} \in G_{in}$}{%
        $e = \Edge{b}{d}$%
      }{%
        $e = \Edge{a}{c}$%
      }
      \leIf{$\Edge{a}{c} \in G_{out}$}{%
        $f = \Edge{b}{d}$%
      }{%
        $f = \Edge{a}{c}$%
      }
     \Return \textbf{merge}($\mathcal{B}(G_{in} + C + e, \eta_{in}),
                    \mathcal{B}(G_{out} + C + f, \eta_{out}))$\;
    }
    \ElseIf{there is a separating triple $C = (a, b, c)$ with \linebreak
         $\mathcal{C}[a, b] = \{\Edge{u}{v}\}$}{
      color $\Edge{a}{b}$ red and $\Edge{u}{v}$ blue (or vice-versa) and let $e$ be the red
      edge
      and $\Edge{b}{c},\Edge{c}{a},\Edge{a}{u}, \Edge{a}{v},
        \Edge{b}{u},$ $\Edge{b}{v}$ black\;
      add $\alpha(\kappa)$ to $\eta$ where $\kappa = G[\{a, b, u, v\}]$\;
     \Return  $\mathcal{B}(G - e, \eta)$\;
    }
    \ElseIf{there is a separating quadruple $C = (a, b, c, d)$ with
        \linebreak
         $\mathcal{C}[a, b] = \{x,y\}$}{
      color $\{a,b\}$ red and $\{ x,y \}$ blue (or vice-versa) and let $e$ be the red
      edge
      and $\Edge{b}{c}, \Edge{c}{d}, \Edge{d}{a},
        \Edge{a}{x}, \Edge{a}{y},$ $\Edge{b}{x}, \Edge{b}{y}$ black\;
      add $\alpha(\kappa)$ to $\eta$ where $\kappa = G[a, b, x, y]$\;
      \Return $\mathcal{B}(G - e, \eta)$\;
    }
    \ElseIf{there is a separating triangle $C = (a, b, c)$ with
        $\mathcal{C}[a, b] = \{\Edge{u}{v}\}$ \linebreak
          and $\mathcal{C}[b, c] = \{\Edge{x}{y}\}$}{
      color one of $\Edge{a}{b}$ and $\Edge{u}{v}$ and one of  $\Edge{b}{c}$  and $\Edge{x}{y}$ red
          and the other ones blue and let $e$ and $f$ be the red edges \linebreak
        and color the edges \linebreak
        $\Edge{c}{a}, \Edge{a}{u}, \Edge{a}{v}, \Edge{b}{u}, \Edge{b}{v}, \Edge{b}{x}, \Edge{b}{y},
        \Edge{c}{x}, \Edge{c}{y}$ black\;
      add $\alpha(\kappa_1) \wedge \alpha(\kappa_2)$ to $\eta$ where $\kappa_1 = G[a, b,
        u, v]$ and $\kappa_2 = G[b, c, x, y]$\;
      \Return  $\mathcal{B}(G - \{e,f\}, \eta)$\;
    }
    \lElseIf{there is a $K_5$}{%
      \Return $(G, \textsf{false})$ and stop
    }
    \lElse{
      \textbf{$MC_4$}($G$, $\eta$)
    }
  }
  \BlankLine
  final-check($G$, $\eta$)\;
  \caption{Algorithm $\mathcal{B}$}
  \label{alg:NIC}
\end{algorithm}

\begin{algorithm}
  \KwIn{A 3-connected graph $G$ with a partial edge coloring
        and a boolean formula $\eta$.}
  \KwOut{A subgraph of $G$, an edge coloring, and $\eta$.}
  \BlankLine
  \If{the detected $K_4$ subgraph $\kappa$ is a completely
    kite-covered tetrahedron}{
    \ForEach{edge $e$ of $\kappa$}{
      color $e = \{a,c\}$ red and the crossing edge $f = \{b,d\}$ of $e$
        blue or vice versa and extend the coloring by black edges $\{a,b\},
        \{b,c\}, \{c,d\}, \{d,a\}$\;
      add a clause $\alpha(\kappa)$ to $\eta$\;
    }
    collect the red edges into a set $F$\;
    \Return $(G-F, \eta)$\;
  }
  \ElseIf{the detected $K_4$ subgraph is $SC$}{
    \ForEach{kite $\kappa$ of $G$ with a pair of crossing edges $e,f$}{%
      color $e$ red and $f$ blue and
        color the remaining edges black\;
     add a clause $\alpha(\kappa)$ to $\eta$\;
    }
    collect the red edges into a set $F$\;
    \Return $(G-F, \eta)$\;
  }
  \Else {color the crossing edges of the detected kite $\kappa$ red and blue
      \linebreak
        \hspace*{0.5cm}
      and color the other edges of $\kappa$ black\;
    add a clause $\alpha(\kappa)$ to $\eta$\;
    \Return $(G - e, \eta)$\ where $e$ is the red edge;
  }
  \caption{Algorithm $MC_4$}
  \label{alg:K4}
\end{algorithm}

\begin{algorithm}
  \KwIn{A 3-connected graph $G$ with a partial edge coloring and a boolean
    formula $\eta$.}
  \KwOut{A planar subgraph of $G$ with an edge coloring and $\eta$.}
  \BlankLine
  \If {$G$ is a triangulated planar graph}{
    extend the coloring by black edges\;
    \Return $(G, \eta)$;
  } \Else {
    \If {$|G| \leq 8$}{
      \If {$G$ has a 1-planar embedding $\mathcal{E}(G)$ extending  the partial edge
           coloring $\gamma$}{
            extend the coloring by grey edges\;
            collect one edge from each pair of crossing edges of
            $\mathcal{E}(G)$ into $F$\;
        \Return $(G-F,   \eta)$\;
      }
    }
    \Return $(G, \textsf{false})$ and \textbf{stop}%
}
  \caption{Algorithm final-check for 1-planar Graphs}
  \label{alg:final-check}
\end{algorithm}

\newpage

The subroutine \emph{merge} reverts the partition into an inner and
an outer subgraph, takes the edge coloring of the subgraphs and
combines the boolean formulas by a conjunction. It ignores the chord
that was added for the triangulation if there is a partition by
a separating  4-cycle.
If $\Edge{x_1}{y_1}, \ldots, \Edge{x_r}{y_r}$ are the crossable
edges of a separating edge $\Edge{a}{b}$, then
$\sigma(a,b,\mathcal{C}[a, b]) = (\alpha_{\kappa_1} \vee \ldots \vee
\alpha_{\kappa_r})$ where $\kappa_i = G[a,b, x_i, y_i ]$ for
$i=1,\ldots, r$ is a $K_4$.

Algorithm final-check takes any 1-planar embedding and makes
progress by removing one edge from each pair of crossing edges.
Uncolored edges are  colored grey although the used embedding
determines pairs of crossing edges. A grey edge may be planar in one
1-planar embedding of $G$ and crossed in another.

The correctness  of algorithms $\mathcal{B}$ uses the following
properties of Algorithm \CGP{}, as proved in
\cite{cgp-rh4mg-06}.

\begin{lemma} \label{lem:correct1}
  If $\mathcal{E}(G)$ is a triangulated 1-planar embedding of $G$, then
  algorithm $\mathcal{B}$ succeeds and returns a triangulated planar spanning
  subgraph, an edge coloring, and a boolean formula $\eta$ such that
  each black edge is planar in $\mathcal{E}(G)$, each red, blue, and orange
  edge is crossed in $\mathcal{E}(G)$, there are pairs of red and blue edges
  that cross, and each orange edge crosses a cyan one. Moreover, one edge
  from each pair of crossing edges is red, orange, or grey.
\end{lemma}

\begin{proof}
  By the assumptions, $G$ is a triangulated 1-planar graph and
  algorithm \CGP{}  succeeds on $G$. Since $\mathcal{B}$ extends
 \CGP{}, it also succeeds, since  $MC_5$ does not apply, as
  shown in Lemma \ref{lem:no-MC5}. Then $\mathcal{B}$ provides the edge coloring
  as stated. All $K_4$ subgraphs are scanned and their embedding is
  classified as a kite, as a planar tetrahedron (after a separating
  3-cycle), or as a tetrahedron with kite-covered edges. One edge of each
  detected kite is colored red or orange, or it is colored grey in a small
  subgraph.
  %
\end{proof}

\section{Specialization to \IC{}- and \NIC{}-Planarity} \label{sect:IC+NIC}

For the test of \IC{}- and \NIC{}-planarity we must specialize
  $MC_4$   and final-check. A completely kite-covered $K_4$
is not \IC{}-planar and the $SC$-graph is not \NIC{}-planar. If such
a subgraph is encountered, then the recognition algorithm  for
\IC{}-planarity (\NIC{}-planarity) fails and the modified algorithm
  returns \textsf{false} and stops.

If $G$ is a small graph of size at most eight, then consider all
\IC{}-planar (\NIC{}-planar) embeddings of $G$ that extend the given
partial coloring. Such graphs occur at separating 3-cycles and
separating 4-cycles and have an outer cycle $C$ whose edges are
colored black, blue, or cyan. If an edge $e$ of $C$ is colored blue,
then $e$ is crossed by an edge $f$, which is colored red and is not
part of the evoked input graph $G$. Then there are black edges in
$G$ which are part of the kite with the pair of crossing edges $\{e,
f\}$. Such situations decrease the number of \IC{}-planar
(\NIC{}-planar) embeddings of $G$ and reduce ambiguities. An
ambiguity in the embeddings is due to   $K_5$, which may choose any
edge of an outer triangle as part of its kite, as shown in
Fig.~\ref{fig:K5fixedface}. In the \IC{}-planar case, we obtain the
formula  $(a_{\kappa_1} \wedge b_{\kappa_1}) \vee (a_{\kappa_2}
\wedge  c_{\kappa_2}) \vee (b_{\kappa_3} \wedge c_{\kappa_3})$,
where $a, b, c$ are the vertices of the outer face and $\kappa_1,
\kappa_2, \kappa_3$ are the three possible kites. For
\NIC{}-planarity be obtain $(e_{\kappa_1} \vee
 f_{\kappa_2} \vee g_{\kappa_3})$, where $e, f$  and $g$ are the edges
of the outer face. Note that  edge $h$ between the two inner
vertices is in each of the three kites, but $h$ is   ignored for the
boolean formula. It is internal and its three occurrences would
violate the special structure of the boolean formula $\eta$ that is
used for an efficient evaluation. Similarly, if $G$ has eight
vertices with vertices $a, b,c, d$ in the outer face (ignoring the
added chord) and internally two $K_5$ including the vertices $a,b,c$
and $b,c,d$, then there are two \IC{}-planar embeddings of $G$  and
the formula is $(a_{\kappa} \wedge b_{\kappa} \wedge c_{\kappa'}
\wedge d_{\kappa'}) \vee (a_{\kappa} \wedge c_{\kappa} \wedge
b_{\kappa'} \wedge d_{\kappa'})$, where $\kappa$ and $\kappa'$ are
two kites. A \NIC{}-planar  embedding uses the diagonal $\{b, c\}$
for one kite and one outer edge for the second kite so that we
obtain the formula $(e_{\kappa_1} \vee f_{\kappa_2}  \vee
g_{\kappa_3} \vee h_{\kappa_4})$ where $e,f,g,h$ are the edges of
the outer face and $\kappa_1, \ldots, \kappa_4$ are possible kites.

The extended version  of algorithm final-check in $\mathcal{B}_{IC}$
($\mathcal{B}_{NIC}$)  computes all \IC{}-planar
(\NIC{}-planar)  embeddings of the small graph $G$ of size at most
eight that extend the given partial coloring. Each embedding is
expressed by a boolean formula in CNF using only variables for the
vertices (edges) in the outer face. These formulas are combined by
a disjunction to express the set of all \IC{}-planar (\NIC{}-planar)
embeddings.

\begin{algorithm}
  \KwIn{A 3-connected graph $G$ with a partial edge coloring and a boolean
    formula $\eta$.}
  \KwOut{A planar subgraph of $G$ with an edge coloring and $\eta$.}
  \BlankLine
  \If {$G$ is a triangulated planar graph}{
    extend the coloring by black edges\;
    \Return $(G, \eta)$;
  } \Else {
    \If {$|G| \leq 8$}{
      \If {$G$ has an \IC{}-planar (\IC{}-planar) embedding extending the partial edge
           coloring $\gamma$}{
     Extend the coloring by grey edges\;
     Set $\beta =$ \textsf{false} for a boolean formula $\beta$\;
     \ForEach{\IC{}-planar (\NIC{}-planar) embedding  $\mathcal{E}(G)$ extending $\gamma$}
            {
            Express the embedding  by a boolean formula $\alpha$
            using only variables for the vertices (edges) in the outer face
            of $\mathcal{E}(G)$ (with the added chord ignored)\;
            Add $\alpha$ by a disjunction to $\beta$\;
      }
        choose any \IC{}-planar (\NIC{}-planar) embedding $\mathcal{E}(G)$ with a set $F$ of crossed
          edges\;
         add $\beta$ by a conjunction to $\eta$\;
        \Return $(G-F,   \eta)$\;
      }
}
    \Return $(G, \textsf{false})$ and stop%
 }
  \caption{Algorithm final-check for \IC{}-planar and \NIC{}-planar Graphs}
  \label{alg:final-check-IC}
\end{algorithm}

\newpage

The boolean formula $\eta$  collects the clause $\alpha(\kappa)$
of all possible kites $\kappa$, but it does not express \IC{}- and
\NIC{}-planarity and a mutual exclusion of two kites with a common
vertex and edge, respectively. Therefore, we extend $\eta$ to
$\eta^+$, and we finally transform $\eta^+$  into $\eta^*$
for an efficient evaluation.

\begin{definition}
  The \emph{\IC{}-extension} (\emph{\NIC{}-extension}) $\eta^+$ of a boolean formula
  $\eta$ is obtained by adding a
  clause
  \begin{equation*}
    \neg x_{\kappa} \vee \neg x_{\kappa'}
  \end{equation*}
  if there is a vertex (edge) $x$ with  two variables $x_{\kappa}$ and $x_{\kappa'}$ and $\kappa  \neq \kappa'$
  in  clauses of $\eta$.
\end{definition}

Let $\mathcal{B}_{IC}$ and $\mathcal{B}_{NIC}$ be the algorithms
obtained from algorithm $\mathcal{B}$ by  the versions for
\IC{}-planar and \NIC{}-planar graphs, respectively.
$\mathcal{B}_{IC}$ returns \textsf{false} and stops  if
$\mathcal{B}$ encounters a separating triangle, if $MC_4$ encounters
a completely kite-covered $K_4$ or an $SC$-graph and it uses
Algorithm 4 for the final test. Accordingly, $\mathcal{B}_{IC}$
returns \textsf{false} and stops  if $MC_4$  encounters an
$SC$-graph and uses Algorithm 4 for the final test.

\begin{theorem} \label{thm:charICplanar}
  A graph $G$ is triangulated \IC{}-planar (\NIC{}-planar) if and only if Algorithm
  $\mathcal{B}_{IC}$ $(\mathcal{B}_{NIC})$ succeeds and returns a  boolean formula $\eta$
  whose \IC{}-extension (\NIC{}-extension) $\eta^+$ is satisfiable.
\end{theorem}

\begin{proof}
By Lemma \ref{lem:correct1}, if algorithm $\mathcal{B}_{IC}$  $(\mathcal{B}_{NIC})$
 succeeds, then $G$ is triangulated 1-planar, since
$\mathcal{B}_{IC}$  $(\mathcal{B}_{NIC})$  specializes algorithm $\mathcal{B}$.
  Construct the embedding $\mathcal{E}(G)$ by an undo of all actions taken by
  $\mathcal{B}$. Then $\mathcal{E}(G)$ is a 1-planar embedding in which all
  black edges are planar and all red, blue, and orange edges are crossed.
  For every $K_4$ subgraph which cannot be
  embedded planar there is a clause expressing a crossing.  Consider a
  valid truth assignment of $\eta^+$. If a clause of $\alpha(\kappa)$ of $\eta$
  resp. $\eta^+$ is satisfied by $x_{\kappa}= \textsf{true}$, then $x_{\kappa'} =
  \textsf{false}$ for all $\kappa'$ such that $\kappa$ and $\kappa'$ share vertex (edge) $e$, which
  implies that the 1-planar embedding $\mathcal{E}(G)$ is \IC{}-planar (\NIC{}-planar).
  Conversely, if $G$ is triangulated \IC{}-planar (\NIC{}-planar), then algorithms $\CGP{}, \mathcal{B}$ and
$\mathcal{B}_{IC}$  $(\mathcal{B}_{NIC})$ succeed. There is an embedding $\mathcal{E}(G)$
that corresponds to the computed   edge coloring.
The  extension $\eta^+$ of the computed boolean formula $\eta$ is satisfied by a truth assignment
$x_{\kappa}= \textsf{true}$ for each kite of  $\mathcal{E}(G)$, since each vertex (edge) $x$ is in
at most one kite.
  %
\end{proof}

Algorithm \CGP{} runs in cubic time \cite{cgp-rh4mg-06}, and so do
$\mathcal{B}, \mathcal{B}_{IC}$ and $\mathcal{B}_{NIC}$. For a cubic
time recognition it remains to show that the satisfiability problem
of $\eta^+$ can be solved in cubic time. This is not immediately
clear, since   the algorithms  collect clauses for the CNF formula $\eta$.
Fortunately,
$\eta$ and $\eta^+$ have a special structure which is used for an
efficient evaluation.

\subsection{Evaluate \IC{}-planar Formulas}

First, consider the \IC{}-planar case.  There is a disjunction at
two places: (i) a separating edge and  (ii) small graphs.

If $e = \{a,b\}$ is a separating edge with candidates $f_1, \ldots,
f_r$ in $\mathcal{C}[a,b]$   for a crossing, then the formula
$\sigma(a,b,\mathcal{C}[a,b])$ expresses all possible kites from
which one is realized. Since the vertices $a$ and $b$ are in the
kite, $\sigma(a,b,\mathcal{C}[a,b])$ can be transformed into
$(a_{\kappa} \wedge b_{\kappa}) \wedge ((x^1_{\kappa} \wedge
y^1_{\kappa}) \vee \ldots  \vee (x^r_{\kappa} \wedge y^r_{\kappa}))$
with $f_i = \{x^i, y^i\}$ for $i=1,\ldots, r$ and $r > 1$.

If $G$ is a small graph after a separating 3- or 4-cycle $C$, then
\IC{}-planar embeddings can be expressed by a 2SAT formula with variables for the vertices of $C$.

\begin{lemma} \label{lem:2SATIC}
If $G$ is an \IC{}-planar graph without a separating edge, then the
\IC{}-extension $\eta^+$  of the formula $\eta$ computed by
algorithm $\mathcal{B}_{IC}$ is equivalent to a $2$SAT formula.
\end{lemma}
\begin{proof}
The boolean formula of a kite $\kappa = G[a,b,c,d]$ is
$\alpha(\kappa) = (a_{\kappa} \wedge  b_{\kappa} \wedge c_{\kappa}
\wedge d_{\kappa})$ and $\alpha(\kappa)$ is added to $\eta$ at a
separating triple, a separating quadruple, and at a kite in $MC_4$.
 Two
subexpressions are combined by a conjunction at a separating 3- and 4-cycle.

It thus remains to show that the \IC{}-extension of a formula $\eta$
for a small graph $H$ is equivalent to a 2SAT formula, and is
replaced by the 2SAT formula for further computations. Graph $H$ is
obtained by a separating 3- or 4-cycle. Then the edges of the cycle
are black, blue or cyan and are treated as planar. Since blue or
cyan imposes further restrictions we suppose they are black. First,
suppose  there is a 3-cycle with vertices $a,b,c$. If there is a
planar embedding, then \textsf{true} is returned. Otherwise, if
every embedding of $G$ contains a kite, then $\eta = (a_{\kappa_1}
\wedge b_{\kappa_1}) \vee (a_{\kappa_2} \wedge c_{\kappa_2}) \vee
(b_{\kappa_3} \wedge c_{\kappa_3})$ is returned if $H$ is $K_5$, see
Fig.~\ref{fig:K5fixedface}. The \IC-extension adds the clauses
$(\neg a_{\kappa_1} \vee \neg a_{\kappa_2}), (\neg b_{\kappa_1} \vee
\neg b_{\kappa_3})$ and $(\neg c_{\kappa_2} \vee \neg
c_{\kappa_3})$. The combined subexpression is equivalent to the 2SAT
formula $(a_{\kappa} \vee b_{\kappa}) \wedge (a_{\kappa} \vee
c_{\kappa}) \wedge (b_{\kappa} \vee c_{\kappa})$  for a new virtual
kite $\kappa$ and is replaced by this formula. If $G$ has 6, 7 or 8
vertices, then $\eta$ is a 2SAT formula, since a single kite
contains at most two of the outer vertices or there are two kites
and 8 vertices and the formula is equivalent to $(a_{\kappa} \wedge
b_{\kappa} \wedge c_{\kappa})$. Similarly, if $H$ is obtained by a
separating 4-cycle $C = (a,b,c,d)$, then there is a $K_5$ with three
outer vertices $a,b,c$ as above, or $\eta$ is a 2SAT formula. In
particular, if $H$ contains two $K_5$ with vertices $a,b,c, u,v$ and
$a,c,d,x,y$, then $\eta = (a_{\kappa_1} \wedge b_{\kappa_1})  \wedge
(c_{\kappa_2} \wedge d_{\kappa_2}) \vee (a_{\kappa_3} \wedge
c_{\kappa_3})  \wedge (b_{\kappa_4} \wedge d_{\kappa_4})$, which
after the \IC{}-extension is equivalent to $(a_{\kappa} \wedge
b_{\kappa} \wedge c_{\kappa} \wedge d_{\kappa})$.
%
\end{proof}

The satisfiability problem for $\eta^+$ can be reduced to a series
of 2SAT satisfiability problems, which altogether can be solved in
linear time in the length of $\eta^+$.
This technique was used in \cite{bdeklm-IC-16}.

\begin{lemma} \label{lem:evalIC}
There is a linear time algorithm (in the length of $\eta^+$) to test the satisfiability of the
\IC{}-extension $\eta^+$ of   a triangulated \IC{}-planar graph.
\end{lemma}
\begin{proof}
Consider the recursive construction of the boolean formula $\eta$ by
algorithm $\mathcal{B}_{IC}$. All subexpressions are in 2SAT form or
are replaced by an equivalent 2SAT formula as shown in Lemma
\ref{lem:2SATIC}, except if there is a separating edge. We proceed
by induction on the depth of the  formula. Consider the first
separating edge $\{a,b\}$ and the computed  formula $\sigma(a,b,
\mathcal{C}[a,b])$. As stated before, $\sigma(a,b,\mathcal{C}[a,b])$
can be transformed into $(a_{\kappa} \wedge b_{\kappa}) \wedge
((x^1_{\kappa} \wedge y^1_{\kappa}) \vee \ldots, \vee (x^r_{\kappa}
\wedge y^r_{\kappa}))$ with $f_i = \{x^i, y^i\}$ for $i=1,\ldots, r$
and $\mathcal{C}[a,b]) = (f_1, \ldots, f_r)$ in this order according
to the rotation system at $a$. Then $a, b, x_1, y_r$ are the outer
vertices, and $x_2, \ldots, x_r, y_1, \ldots, y_{r-1}$ are inner
vertices. We wish to use a pair of inner vertices (together with $a$
and $b$) for the kite. For $i=2, \ldots, r-1$ check the
satisfiability of the subexpression $\eta^+(a,x_i, b, y_i)$ of
$\eta^+$ that corresponds to the subgraph $H$ of $G$ that is
separated by the 4-cycle $C = (a, x_i, b, y_i)$. Since $H$ has no
separating edge, $\eta^+(a,x_i, b, y_i)$ is equivalent to a 2SAT
formula, whose satisfiability is checked in linear time. If
$\eta^+(a,x_i, b, y_i)$ is satisfiable for some $i$, then replace
$\sigma(a,b,\mathcal{C}[a,b])$ by $(a_{\kappa} \wedge b_{\kappa})$.
Otherwise, check the satisfiability of $\eta^+(a,x_1, b, y_1)$ and
of $\eta^+(a,x_r, b, y_r)$. If neither of them is satisfiable, then
$\eta^+$ is not satisfiable and the given graph $G$ is not
\IC{}-planar. If only one is satisfiable, say $\eta^+(a,x^1, b,
y^1)$, then replace $\sigma(a,b,\mathcal{C}[a,b])$ by $(a_{\kappa}
\wedge b_{\kappa} \wedge x^1_{\,\kappa})$, and replace
$\sigma(a,b,\mathcal{C}[a,b])$ by $(a_{\kappa} \wedge b_{\kappa})
\wedge (x^1_{\,\kappa} \vee x^r_{\,\kappa})$. Each subexpression
$\eta^+(a,x_i, b, y_i)$ is evaluated at most once, namely for the
satisfiability check of $\sigma(a,b,\mathcal{C}[a,b])$, and
$\sigma(a,b,\mathcal{C}[a,b])$ is replaced by a short formula in the
further computation of the satisfiability check of $\eta^+$.
%
\end{proof}

A 1-planar graph of size $n$ has at most $n-2$ kites, and the bound
is achieved by optimal 1-planar graphs with $4n-8$ edges
\cite{bsw-1og-84, s-s1pg-86}. Hence, the boolean formula $\eta$ has
length $O(n)$. However, the \IC{}-extension may add up to $O(n^2)$
subexpressions of the form     $(\neg x_{\kappa}  \vee \neg
x_{\kappa'})$, for example, if $x$ is the center of a star of $K_5$,
as illustrated in Fig.~\ref{fig:multiK5}.

\begin{figure}
    \begin{center}
      \includegraphics[scale=0.4]{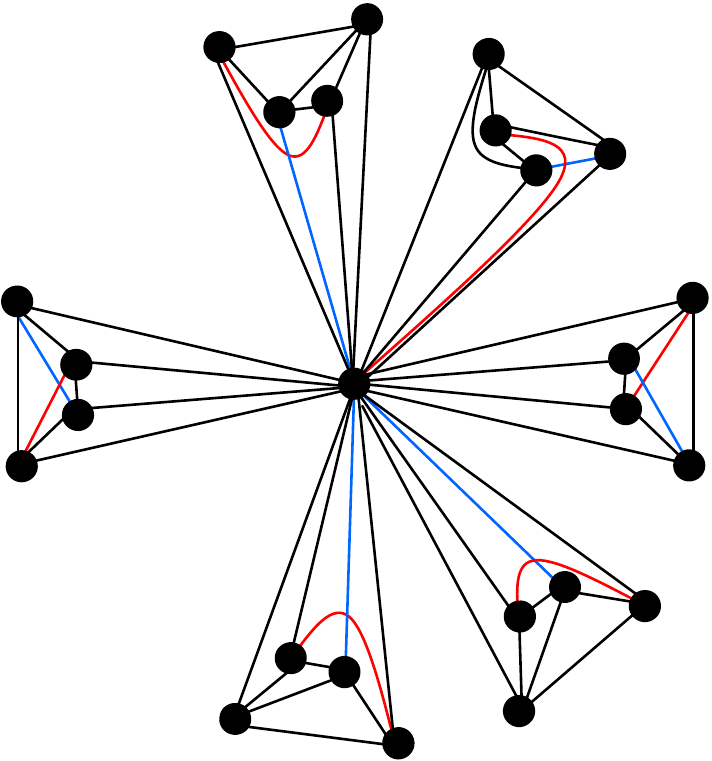}
      \caption{A graph with many $K_5$ subgraphs inducing many \IC{}-extensions }
      \label{fig:multiK5}
    \end{center}
 \end{figure}

In consequence, the satisfiability check of $\eta^+$ takes linear
time in the length of the formula and at most quadratic time in the
size of the input graph. This is less than the cubic running time of
algorithm $\mathcal{B}_{IC}$.

In summary, we obtain:

\begin{theorem} \label{thm:ICmain}
Triangulated \IC{}-planar graphs can be recognized in cubic time.
\end{theorem}

\subsection{Evaluate \NIC{}-planar formulas}


Algorithm $\mathcal{B}_{NIC}$ extends algorithm $\mathcal{B}$ and
stops with a failure if  $MC_4$    detects an $SC$-graph. In
addition, it computes all \NIC{}-planar embeddings of small graphs
and expresses them by a boolean formula $\eta$. A boolean variable
has the form $e_{\kappa}$ for some edge $e$ and a kite $\kappa$
containing $e$ as a planar edge. If $G$ is a small graph that
results from a separating 3- or 4-cycle $C$, then only the edges of
$C$ are taken into account. An inner edge of a $K_5$ occurs in
 three kites and does not satisfy the next lemma.

\begin{lemma} \label{lem:NIC2occurrences}
For every edge $e$ with an occurrence of a variable  $e_{\kappa}$ in
$\eta$ there are at most two variables $e_{\kappa}$ and
$e_{\kappa'}$ in $\eta$.
\end{lemma}

\begin{proof}
First, observe that each  variable $e_{\kappa}$ occurs once in $\eta$
since algorithms \CGP{} and  $\mathcal{B}_{NIC}$ check each $K_4$
exactly once and $\mathcal{B}_{NIC}$ introduces a boolean variable
 $e_{\kappa}$ if edge $e$ is a planar edge of candidate kite $\kappa$,
which is a unique event.

Second, all steps of algorithm $\mathcal{B}_{NIC}$ introduce at most two variables
$e_{\kappa}$ and $e_{\kappa'}$ for two candidate kites $\kappa$ and $\kappa'$,
except if $G$ is a small graph of size at most eight. Let $C$ be the
edges of the outer cycle of $G$. Then $C$ is a 3-cycle or a
$4$-cycle and $G$ and $C$ are obtained from the separating 3-cycle
and 4-cycle steps of the algorithm, except if the input graph is
small, which is checked by exhaustive search.
Suppose that the edges of $C$ are colored black. Otherwise,
 the possible \NIC{}-planar
embeddings of $G$ are more restricted.

Assume that $G$ must be embedded with at least one kite. First,
suppose that $C$ is a 3-cycle with edges $e^1, e^2$ and $ e^3$. If
$G = K_5$, then exactly one edge $e$ of $C$ is part of a kite, as
illustrated in Fig.~\ref{fig:K5fixedface}, and all \NIC{}-planar
embeddings are expressed by the boolean formula $\alpha =
e^1_{\kappa_1} \vee e^2_{\kappa_2} \vee e^3_{\kappa_3}$.
 A graph $G$
of size six with outer cycle $C$ is \NIC{}-planar if there is a
fixed kite including one edge of $C$ or there is a $K_5$ with one of
two edges of $C$. Hence, $\alpha = e_{\kappa}$ or $\alpha =
e_{\kappa}  \vee e'_{\kappa'}$ for edges $e, e' \in \{e^1, e^2,
e^3\}$. If $G$ has size seven, then a \NIC{}-planar embedding  has
up to two kites. At least one  kite is fixed and the other is part
of  a $K_5$, see Fig.~\ref{fig:G7}.  The possible \NIC{}-planar
embeddings are expressed by $\alpha = e(\kappa) \wedge e'(\kappa')$
if both $K_4$ are fixed and by   $\alpha = e_1(\kappa_1) \wedge
(e_2(\kappa_2) \vee e_3(\kappa_3))$, otherwise. Similarly, if there
is only one kite that needs an edge of $C$, then $\alpha =
e(\kappa)$ or $\alpha =  e(\kappa) \vee e'(\kappa')$ express  these
embeddings.
Finally, if $G$ has size eight, then it has up to two fixed $K_4$ or
two outer edges may belong to a $K_5$ and the other edge to a fixed
$K_4$ such that $\alpha = e^1(\kappa_1) \wedge e^2(\kappa_2)$ or
$\alpha = e^1(\kappa_1) \wedge (e^2(\kappa_2) \vee e^3(\kappa_3))$
or $\alpha$ is a subexpression thereof.

Similarly, if $C$ is a 4-cycle with edges $e^1, e^2, e^3, e^4$, then
a graph of size eight can host two $K_5$ and at most two kites. Each
$K_5$ indices a kite and one of the two $K_5$ shall form a kite
using only inner edges including the diagonal, see
Fig.~\ref{fig:G8}. These \NIC{}-planar embeddings are expressed by
$\alpha = (e^1(\kappa_1) \vee e^2(\kappa_2) \vee e^3(\kappa_3)\vee
e^4(\kappa_4))$. Clearly both $K_4$ can be fixed using up to three
edges of $C$. Each option is  expressed by a boolean formula that
uses a variable $e_{\kappa}$ for each edge $e$ of $C$  at most once.
%
\end{proof}

\begin{figure}
  \centering
  \subfigure[] {
     \includegraphics[scale=0.33]{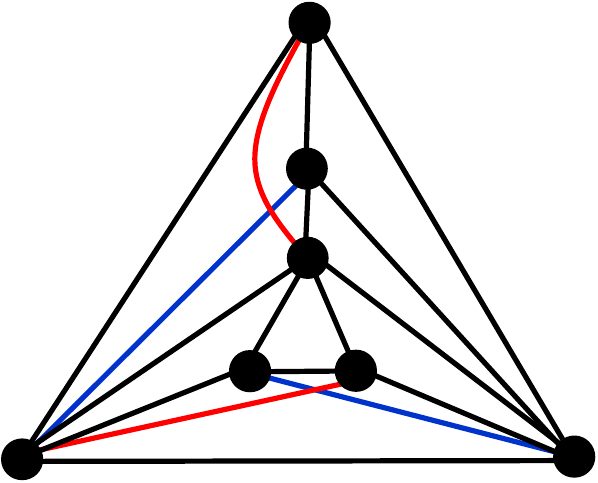}
    \label{fig:G7a}
  }
  \hfil
  \subfigure[] {
      \includegraphics[scale=0.33]{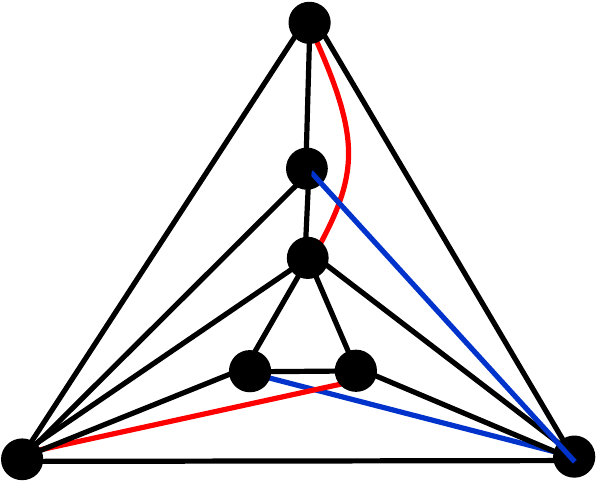}
      \label{fig:G7b}
  }
  \caption{\NIC{}-planar embeddings of a small graph with 7 vertices}
  \label{fig:G7}
\end{figure}

\begin{figure}
  \centering
  \subfigure[] {
     \includegraphics[scale=0.33]{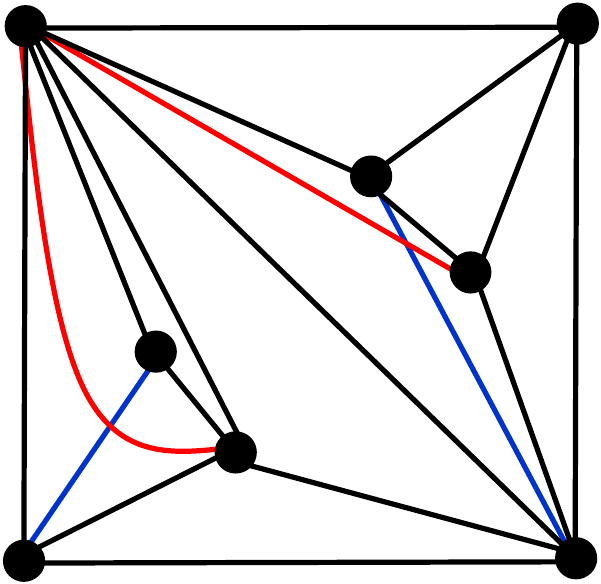}
    \label{fig:G8a}
  }
  \hfil
  \subfigure[] {
      \includegraphics[scale=0.33]{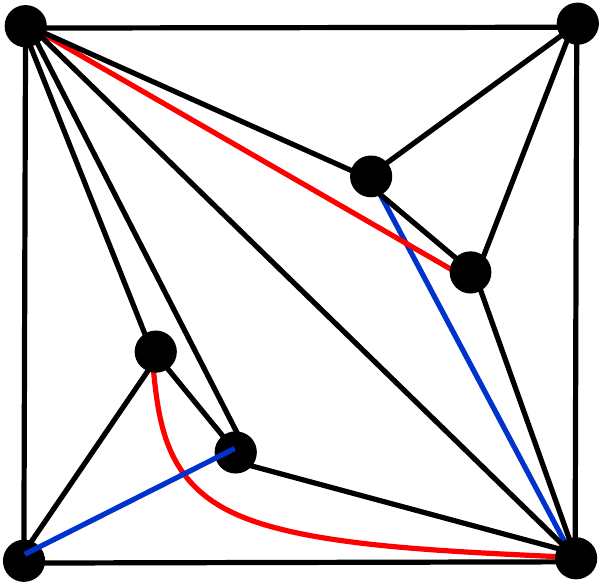}
      \label{fig:G8b}
  }
  \caption{\NIC{}-planar embeddings of a small graph with 8 vertices, a planar 4-cycle and two inner $K_5$}
  \label{fig:G8}
\end{figure}

\begin{lemma} \label{lem:NICcheck}
  The satisfiability of $\eta^+$ can be checked in linear time.
\end{lemma}

\begin{proof}
Consider the formulas $\eta$ and $\eta^+$ that are obtained by
algorithm $\mathcal{B}_{NIC}$ and the \NIC{}-extension.

Set a variable $e(\kappa) = \textsf{true}$ if there is a single
variable $e(\kappa)$ for some edge $e$ in $\eta$. Then there is no
\NIC-extension. Simplify $\eta$ and $\eta^+$ accordingly. Otherwise,
there are  two variables $e(\kappa)$ and $e(\kappa')$ with $\kappa
\neq \kappa'$ in $\eta$ and there is the \NIC{}-extension $\neg
e(\kappa) \vee \neg e(\kappa')$. Keep the first occurrence
$e(\kappa)$ in $\eta$ and replace the second occurrence $e(\kappa')$
by $\neg e(\kappa)$ and omit the \NIC{}-extensions. Since each edge
$e$ occurs in  two variables $e(\kappa)$ and $e(\kappa')$, this
transformation preserves the equivalence of the boolean formulas.
Let $\eta^*$ denote the resulting formula. Then $\eta^*$ is
satisfiable if and only if there is no complementary pair $(x, \neg
x)$ in some clause, which can be checked in linear time.  Each
simplification of $\eta$ can be done in constant time. Hence, the
satisfiability test of $\eta, \eta^+$ and $\eta^*$ takes   linear
time.
%
\end{proof}

As before, there are at most $O(n)$ kites such that the length of
$\eta$ is linear in the size of the given graph $G$. Each edge $e$
belongs to at most two kites. Thus, the \NIC{}-extensions adds at
most $O(n)$ 2SAT clauses. By Lemma \ref{lem:NICcheck}   the
satisfiability check of $\eta^+$ takes linear time in the size of
the graph. This is dominated by the running time of algorithm
$\mathcal{B}_{NIC}$, and we can conclude:

\begin{theorem} \label{thm:NICmainresult}
Triangulated \NIC{}-planar graphs can be recognized in cubic time.
\end{theorem}

\subsection{Embeddings}

We can generalize Thms. \ref{thm:ICmain} and \ref{thm:NICmainresult}
to  graphs $G$ whose 3-connected components are triangulated
\IC{}-planar and \NIC{}-planar. Consider a separation pair $\{u,v\}$
of $G$ and components $H_1,\ldots, H_r$ for some $r > 1$ which each
contain the vertices $u$ and $v$ and the edge $e = \{u,v\}$. Then
$G$ is \IC{}-planar  if and only if each $H_i$ for $i=1,\ldots,r$ is
\IC{}-planar and each of $u$ and $v$  occurs in at most one kite.
This is checked as follows: Let algorithm $\mathcal{B}_{IC}$ return
the boolean formula $\eta_i$ on $H_i$. If $\eta_i$ has a variable
$x(\kappa)$ for $x \in \{u,v\}$ and some kite $\kappa$, then set
$x(\kappa) = \textsf{false}$. If thereafter the \IC{}-extension
$\eta_i^+$ is not satisfiable, then vertex $x$
  is needed in $H_i$. Graph $G$ is not \IC{}-planar if
there are at least two such $H_i$. Accordingly, if $v$ is an
articulation vertex with components $J_1, \ldots, J_s$ for some $s >
1$, then check each component $J_i+v$ and check that $v$ is in a
kite of at most one component. Clearly, a graph is \IC{}-planar if
so are its disconnected components.

Similarly, $G$ is \NIC{}-planar  if and only if each 2-connected
component $H_i$   is \NIC{}-planar and edge $e$ occurs in at most
one kite, which is checked as before.

Note that the decomposition of a graph at its separation pairs
corresponds to the introduction of holes in maps \cite{b-4m1pg-15}.
We can summarize:

\begin{theorem} \label{thm:2connected}
  There is a cubic-time algorithm to test whether a graph is
  \IC{}-planar  (\NIC{}-planar) if each 3-connected component is triangulated
  1-planar.
\end{theorem}

The algorithms $\mathcal{B}$, $\mathcal{B}_{IC}$ and
$\mathcal{B}_{IC}$  color  the edges of a triangulated 1-planar
graph such that  black edges are always embedded planar, red, blue
and orange edges are  crossing edges, and the status of cyan and
grey edges is open. Nevertheless, the algorithms cannot determine an
embedding. There are
 triangulated \IC{}-planar and \NIC{}-planar graphs with exponentially many embeddings, as
illustrated in Figs. \ref{fig:IC-ambig} and \ref{fig:NIC-ambig}.
These graphs are maximal in their class. If $\mathcal{B}$ is applied
to these graphs, it finds separating 3-cycles and many small graphs
with a $K_5$, which each allow for at least two embeddings.

\begin{figure}
   \begin{center}
     \rotatebox{270}{\includegraphics[scale=0.4]{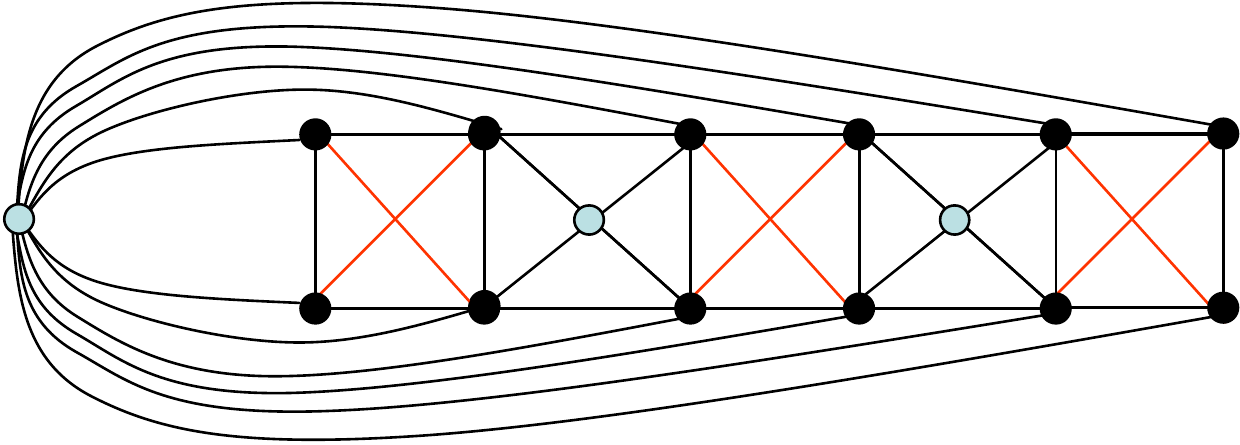}}
     \caption{An IC-planar graph with many embeddings. Each 5-wheel with an open inner vertex
        can be flipped, which implies a
       change of planar and crossing edges in the adjacent kites. Such flips do not change the picture.}
     \label{fig:IC-ambig}
   \end{center}
\end{figure}

 \begin{figure}
    \begin{center}
     \includegraphics[scale=0.4]{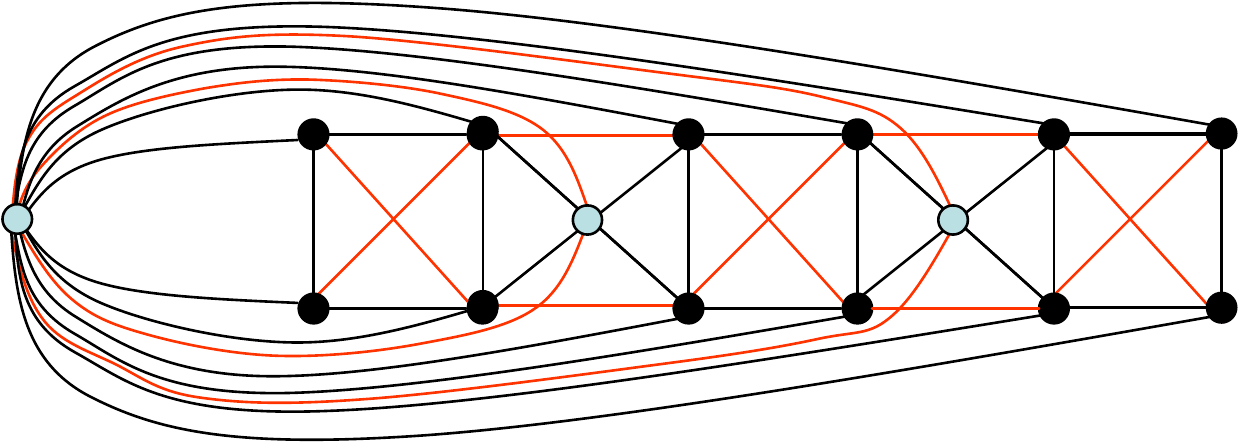}
      \caption{A (maximal) \NIC{}-planar graph $G$ with many embeddings, which are due to the
      separating 3-cycles. Each 5-wheel with an
         open inner vertex can be flipped, which implies a change of planar
         and crossing edges in the adjacent kites. Such flips do not change
         the picture. }
      \label{fig:NIC-ambig}
    \end{center}
 \end{figure}

\subsection{Maximal Graphs}

Obviously, we can test maximality by exhaustive search on graphs
$G+e$, such that $G$ is 1-planar and $G+e$ is not. In the
planar-maximal case, edge $e$ is colored black, and therefore must
be embedded planar.

\begin{theorem}
  For a graph $G$ it takes $\mathcal{O}(n^5)$ time to test whether $G$ is
  planar-maximal and maximal 1-planar, \IC{}-planar, and \NIC{}-planar, respectively.
\end{theorem}

There are special linear time algorithm for optimal 1-planar graphs
\cite{b-rso1p-16} and for \NIC{}-planar graphs \cite{bbhnr-NIC-16}
which use the particular structure of optimal graphs. Recall that
there are optimal   \IC{}-planar (\NIC{}-planar) graphs only for
values  $n=4k+4$ ($n=5k+2$) and $k \geq 1$.

\begin{theorem}
  For a graph $G$ it takes $\mathcal{O}(n^3)$ time to test whether $G$ is
  maximum  (optimal)   \IC{}-planar (\NIC{}-planar).
\end{theorem}
\begin{proof}
 A graph $G$ is maximum (or densest) \IC{}-planar if $G$ has  $\lfloor 3.25 \, n -6 \rfloor$
 edges and is triangulated \IC{}-planar.
Similarly, test whether $G$ is triangulated \NIC{}-planar and has
the maximum number of edges of edges of \NIC{}-planar graphs of size
$b$, which is  $d(n) = \lfloor  3.6 \, (n-2)  \rfloor - \epsilon$
where $\epsilon = 0,1$ and $\epsilon = 0$ for $n=5k+i$ and $i=2,3$
and $k\geq 2$. The  value of $epsion$ for all $n$ is not yet known.
\end{proof}

\subsection{\NP-Completeness}

Finally, we improve upon the \NP-hardness proofs of 1-planarity.
We prove that for every instance $\alpha$ of  planar 3-SAT there
is a graph $G_{\alpha}$ such that $G_{\alpha}$ is \IC{} planar if $\alpha$  is satisfiable
and $\alpha$  is satisfiable if $G_{\alpha}$ is 1-planar, even if  $G_{\alpha}$ is 3-connected and
is given with a rotation system. In consequence, we obtain:

\begin{theorem}
  For a graph $G$ it is \NP-complete to test whether $G$ is \IC{}-planar (\NIC{} planar),
  even
  if $G$ is 3-connected and is given with (or without) a rotation system.
\end{theorem}

\begin{proof}
  We combine and generalize the \NP-hardness results of Auer et al.
  \cite{abgr-1prs-15} on 3-connected 1-planar graphs and of Brandenburg et al.
  \cite{bdeklm-IC-16}.
  Both approaches reduce from planar 3-SAT \cite{l-pftu-82} and use gadgets with a
  unique embedding and the membrane technique.

  We replace the $U$-graphs from \cite{abgr-1prs-15} by $M^+$-graphs which are
   modified grid graphs with a sequence of free connection vertices
  in the outer face, and are drawn as circles in Fig. \ref{fig:M-graph}. Our $M^+$-graphs
  extend the $M$-graphs of \cite{bdeklm-IC-16} by further edges which yield a triangulation
  except for the side with the connection vertices. Obviously, $M^+$-graphs are \IC{}-planar,
  and even more, $M^+$-graphs have a unique 1-planar embedding if the
  connection vertices are in the outer face. The uniqueness is obtained
  by using Algorithms \CGP{} or  $\mathcal{B}$ on an extension of an $M^+$-graph to a triangulated graph
  $N$ which has a
  new vertex $z$ in the outer face so that $z$ is connected with each connection vertex.
Only   $MC_4$   with a kite applies to $N$. Hence, $N$ has a unique 1-planar embedding.
%
%
%
Now we can use the gadgets from the \NP-hardness proof of
\cite{bdeklm-IC-16} with $M^+$ graphs instead of $M$-graphs.
For every planar 3-SAT instance $\alpha$ construct a graph $G_{\alpha}$.
Then $G_{\alpha}$ is 3-connected IC-planar if the instance of planar 3-SAT is satisfiable.
 The rotation system
  can directly be derived from the embedding, and it is unique.
 Conversely, if $G_{\alpha}$ is 1-planar, then it is \IC{}-planar and $\alpha$ is satisfiable.
 %
\end{proof}

\begin{figure}
   \begin{center}
     \rotatebox{270}{\includegraphics[scale=0.35]{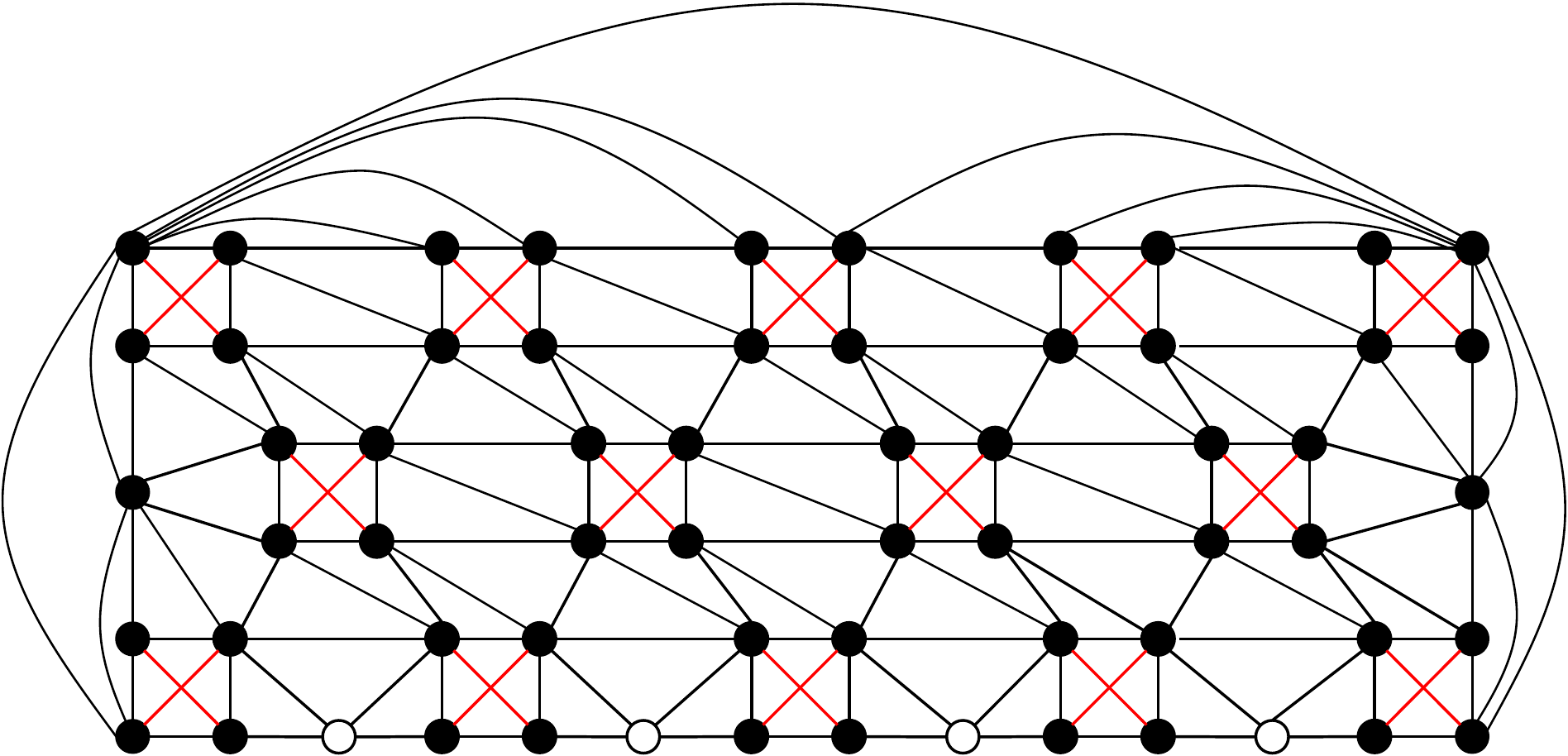}}
     \caption{The structure of an $M^+$ graph.}
     \label{fig:M-graph}
   \end{center}
\end{figure}

\section{Conclusion}
\label{sect:conclusion}

We have shown that triangulated, (planar) maximal, maximum and
optimal \IC{}-planar  and \NIC{}-planar graphs can be recognized in
at most $O(n^5)$ time. Similar bounds are known for 1-planar graphs.
On the other hand, the recognition problem remains \NP-hard if the
graphs are 3-connected. However, the complexity of the recognition
problem for 4- and 5-connected  \IC{}-planar  and \NIC{}-planar
graphs remains open.

\bibliographystyle{splncs03}
\bibliography{brandybibLNCSV4}
\end{document}